\documentclass{llncs}

\usepackage{amsmath,amssymb,amsfonts,mathrsfs,stmaryrd,comment}
\usepackage{bussproofs}
\usepackage{multirow}

\usepackage{url}

\usepackage{wrapfig,blindtext}

\usepackage{fancyhdr}
\usepackage{pbox}

\usepackage{graphicx}
\usepackage{color,soul}

\usepackage{textcomp}

\usepackage{graphs}
 \graphnodesize{0.55}
 \grapharrowlength{0.2}
 \grapharrowwidth{0.6}
 \autodistance{1.0}
 \graphnodecolour{1}

\newcommand{\fillednodes}{\fillednodestrue \graphnodecolour{0} \graphnodesize{0.125} \graphlinewidth{0.015} \grapharrowwidth{0.6}}

\newcommand{\G}{\mathcal{G}}
\newcommand{\R}{\mathcal{R}}

\newcommand{\dder}{\Rightarrow}

\newcommand{\mt}[1]{\text{\tt{#1}}}

\newcommand{\tuple}[1]{\langle#1\rangle}
\newcommand{\failrm}{\mathrm{fail}}

\newcommand{\ifte}[3]{\mathtt{if}\ #1\ \mathtt{then}\ #2\ \mathtt{else}\ #3}

\newcommand{\tryte}[3]{\mathtt{try}\ #1\ \mathtt{then}\ #2\ \mathtt{else}\ #3}

\begin{document}
	
	\title{Verifying Monadic Second-Order Properties of Graph Programs}
	\subtitle{\vspace{10pt}\emph{---\ extended version\ ---}}
	\author{Christopher M. Poskitt\inst{1} and Detlef Plump\inst{2}}

	\institute{Department of Computer Science, ETH Z\"{u}rich, Switzerland
	\and Department of Computer Science, The University of York, UK
	}

	\maketitle
	
	\begin{center}\textbf{Updated:} 16th June 2014\end{center}

	\begin{abstract}
		The core challenge in a Hoare- or Dijkstra-style proof system for graph programs is in defining a weakest liberal precondition construction with respect to a rule and a postcondition. Previous work addressing this has focused on assertion languages for first-order properties, which are unable to express important global properties of graphs such as acyclicity, connectedness, or existence of paths. In this paper, we extend the nested graph conditions of Habel, Pennemann, and Rensink to make them equivalently expressive to monadic second-order logic on graphs. We present a weakest liberal precondition construction for these assertions, and demonstrate its use in verifying non-local correctness specifications of graph programs in the sense of Habel et al.
	\end{abstract}

	\section{Introduction}
	
	Many problems in computer science and software engineering can be modelled in terms of graphs and graph transformation, including the specification and analysis of pointer structures, object-oriented systems, and model transformations; to name just a few. These applications, amongst others, motivate the development of techniques for verifying the functional correctness of both graph transformation rules and programs constructed over them.
	
	A recent strand of research along these lines has resulted in the development of \emph{proof calculi} for graph programs. These, in general, provide a means of systematically proving that a program is correct relative to a specification. A first approach was considered by Habel, Pennemann, and Rensink \cite{Habel-Pennemann-Rensink06a,Pennemann09a}, who contributed weakest precondition calculi -- in the style of Dijkstra -- for simple rule-based programs, with specifications expressed using \emph{nested conditions} \cite{Habel-Pennemann09a}. Subsequently, we developed Hoare logics \cite{Poskitt-Plump12a,Poskitt13a} for the graph transformation language GP 2 \cite{Plump12a}, which additionally allows computation over labels, and employed as a specification language an extension of nested conditions with support for expressions.
	
	Both approaches suffer from a common drawback, in that they are limited to first-order structural properties. In particular, neither of them support proofs about important \emph{non-local} properties of graphs, e.g.\ acyclicity, connectedness, or the existence of arbitrary-length paths. Part of the difficulty in supporting such assertions is at the core of both approaches: defining an effective construction for the weakest property guaranteeing that an application of a given rule will establish a given postcondition (i.e.\ the construction of a \emph{weakest liberal precondition} for graph transformation rules).
	
	Our paper addresses exactly this challenge. We define an extension of nested conditions that is equivalently expressive to monadic second-order (MSO) logic on graphs \cite{Courcelle-Engelfriet12a}. For this assertion language, and for graph programs similar to those of \cite{Habel-Pennemann-Rensink06a,Pennemann09a}, we define a weakest liberal precondition construction that can be integrated into Dijkstra- and Hoare-style proof calculi. Finally we demonstrate its use in verifying non-local correctness specifications (properties including that the graph is bipartite, acyclic) of some simple programs.
	
	The paper is organised as follows. In Section \ref{sec:preliminaries} we provide some preliminary definitions and notations. In Section \ref{sec:expressing_mso_properties} we define an extension of nested conditions for MSO properties. In Section \ref{sec:graph_programs} we define graph programs, before presenting our weakest liberal precondition construction in Section \ref{sec:constructing_wlps}, and demonstrating in Section \ref{sec:example_proofs} its use in Hoare-style correctness proofs. Finally, Section \ref{sec:related_work} presents some related work before we conclude the paper in Section \ref{sec:conclusion}.
	
	This is an extended version of \cite{Poskitt-Plump14a}, and includes the semantics of graph programs as well as the missing proofs.
	
	\section{Preliminaries}\label{sec:preliminaries}
	
	Let $\mathbb{B} = \{\text{true},\text{false}\}$ denote the set of Boolean values, $\text{Vertex},\text{Edge}$ denote (disjoint) sets of node and edge identifiers (which shall be written in lowercase typewriter font, e.g.\ $\mathtt{v},\mathtt{e}$), and $\text{VSetVar},\text{ESetVar}$ denote (disjoint) sets of node- and edge-set variables (which shall be written in uppercase typewriter font, e.g.\ $\mathtt{X},\mathtt{Y}$).
	
	A \emph{graph}\/ over a label alphabet $\mathcal{C} = \langle \mathcal{C}_V, \mathcal{C}_E \rangle$ is defined as a system $G=(V_G,E_G,s_G,t_G,l_G,m_G)$, where $V_G\subset\text{Vertex}$ and $E_G\subset\text{Edge}$ are finite sets of \emph{nodes} (or \emph{vertices}) and \emph{edges}, $s_G,t_G\colon E_G\rightarrow V_G$ are the \emph{source} and \emph{target} functions for edges, $l_G\colon V_G\to \mathcal{C}_V$ is the node labelling function and $m_G\colon E_G\to \mathcal{C}_E$ is the edge labelling function. The \emph{empty graph}, denoted by $\emptyset$, has empty node and edge sets. For simplicity, we fix the label alphabet throughout this paper as $\mathcal{L} = \langle \{\square\}, \{\square\} \rangle$, where $\square$ denotes the blank label (which we render as $\begin{graph}(0.3,0.1)(0,0)  \fillednodes \roundnode{A}(0.2,0.075) \end{graph}$ and $\begin{graph}(0.8,0.1)(0,0) \fillednodes \graphnodecolour{1} \graphnodesize{0}  \roundnode{A}(0.1,0.075)  \roundnode{B}(0.7,0.075)  \diredge{A}{B} 
	 \end{graph}$ in pictures). We note that our technical results hold for any fixed finite label alphabet.
	
	Given a graph $G$, the \emph{(directed) path predicate} $\text{path}_G\!: V_G \times V_G \times 2^{E_G} \rightarrow \mathbb{B}$ is defined inductively for nodes $v,w\in V_G$ and sets of edges $E\subseteq E_G$. If $v=w$, then $\text{path}_G(v,w,E)$ holds. If $v\neq w$, then $\text{path}_G(v,w,E)$ holds if there exists an edge $e\in E_G\setminus E$ such that $s_G(e) = v$ and $\text{path}_G(t_G(e),w,E)$.

	A \emph{graph morphism} $g\colon G\rightarrow H$ between graphs $G,H$ consists of two functions $g_V\colon V_G\rightarrow V_H$ and $g_E\colon E_G\rightarrow E_H$\/ that preserve sources, targets and labels; that is, $s_H\circ g_E=g_V\circ s_G$, $t_H\circ g_E=g_V\circ t_G$, $l_H \circ g_V = l_G$, and $m_H \circ g_E = m_G$. We call $G,H$ the \emph{domain} (resp.\ \emph{codomain}) of $g$. Morphism $g$ is an \emph{inclusion} if $g(x)=x$ for all nodes and edges $x$. It is \emph{injective}\/ (\emph{surjective}) if $g_V$ and $g_E$ are injective (surjective). It is an \emph{isomorphism} if it is both injective and surjective. In this case $G$ and $H$\/ are \emph{isomorphic}, which is denoted by $G\cong H$.

	\section{Expressing Monadic Second-Order Properties}\label{sec:expressing_mso_properties}
	
	We extend the nested conditions of \cite{Habel-Pennemann09a} to a formalism equivalently expressive to MSO logic on graphs. The idea is to introduce new quantifiers for node- and edge-set variables, and equip morphisms with constraints about set membership. The definition of satisfaction is then extended to require an interpretation of these variables in the graph such that the constraint evaluates to true. Furthermore, constraints can also make use of a predicate for explicitly expressing properties about directed paths. Such properties can of course be expressed in terms of MSO expressions, but the predicate is provided as a more compact alternative.

	\begin{definition}[Interpretation; interpretation constraint]\label{def:interpretation}\rm
		Given a graph $G$, an \emph{interpretation $I$ in $G$} is a partial function $I\!:\text{VSetVar}\cup\text{ESetVar}\rightarrow2^{V_G} \cup 2^{E_G}$, such that for all variables $\mathtt{X}$ on which it is defined, $I(\mathtt{X}) \in 2^{V_G}$ if $\mathtt{X}\in \text{VSetVar}$ (resp.\ $2^{E_G}$, ESetVar). An \emph{(interpretation) constraint} is a Boolean expression that can be derived from the syntactic category Constraint of the following grammar:

	\begin{center}
			\begin{tabular}{lcl}
		Constraint & ::= & Vertex$\ \text{'}\mathtt{\in}\text{'}\ $VSetVar  $\mid$ Edge$\ \text{'}\mathtt{\in}\text{'}\ $ESetVar\\
					& & $\mid$ \verb#path# '\verb#(#' Vertex '\verb#,#' Vertex ['\verb#,#' $\mathtt{not}$ Edge \{'\verb#|#' Edge\}] '\verb#)#'\\
					& & $\mid$ $\mathtt{not}$ Constraint $\mid$ Constraint ($\mathtt{and}$ $\mid$ $\mathtt{or}$) Constraint $\mid$ $\mathtt{true}$ \\
			\end{tabular}
	\end{center}

	Given a constraint $\gamma$, an interpretation $I$ in $G$, and a morphism $q$ with codomain $G$, the value of $\gamma^{I,q}$ in $\mathbb{B}$ is defined inductively. If $\gamma$ contains a set variable for which $I$ is undefined, then $\gamma^{I,q} = \text{false}$. Otherwise, if $\gamma$ is $\mathtt{true}$, then $\gamma^{I,q} = \text{true}$. If $\gamma$ has the form $x\in\mathtt{X}$ with $x$ a node or edge identifier and $\mathtt{X}$ a set variable, then $\gamma^{I,q}=\text{true}$ if $q(x) \in I(\mathtt{X})$. If $\gamma$ has the form $\mathtt{path}\mt{(}v\mt{,}w\mt{)}$ with $v,w$ node identifiers, then $\gamma^{I,q}=\text{true}$ if the predicate $\text{path}_G(q(v),q(w),\emptyset)$ holds. If $\gamma$ has the form $\mathtt{path}\mt{(}v\mt{,}w\mt{,}\mathtt{not}\ e_1\mt{|}\dots\mt{|}e_n\mt{)}$ with $v,w$ node identifiers and $e_1,\dots,e_n$ edge identifiers, then $\gamma^{I,q}=\text{true}$ if it is the case that the path predicate $\text{path}_G(q(v),q(w),\{q(e_1),\dots,q(e_n)\})$ holds. If $\gamma$ has the form $\mathtt{not}\ \gamma_1$ with $\gamma_1$ a constraint, then $\gamma^{I,q} = \text{true}$ if $\gamma_1^{I,q} = \text{false}$. If $\gamma$ has the form $\gamma_1\ \mathtt{and}\ \gamma_2$ (resp.\ $\gamma_1\ \mathtt{or}\ \gamma_2$) with $\gamma_1,\gamma_2$ constraints, then $\gamma^{I,q} = \text{true}$ if both (resp.\ at least one of) $\gamma_1^{I,q}$ and $\gamma_2^{I,q}$ evaluate(s) to true.
		\qed
	\end{definition}

	\begin{definition}[M-condition; M-constraint]\rm\label{defn:M-condition}
		An \emph{MSO condition} (short.\ \emph{M-condition}) over a graph $P$ is of the form $\mathtt{true}$, $\exists_\mathtt{V} \mathtt{X}\mt{[}c\mt{]}$, $\exists_\mathtt{E} \mathtt{X}\mt{[}c\mt{]}$, or $\exists (a \mid \gamma, c')$, where $\mathtt{X}\in\text{VSetVar}$ (resp.\ ESetVar), $c$ is an M-condition over $P$, $a\!: P \hookrightarrow C$ is an injective morphism (since we consider programs with injective matching), $\gamma$ is an interpretation constraint over items in $C$, and $c'$ is an M-condition over $C$. Furthermore, Boolean formulae over M-conditions over $P$ are also M-conditions over $P$; that is, $\neg c$, $c_1 \wedge c_2$, and $c_1 \vee c_2$ are M-conditions over $P$ if $c,c_1,c_2$ are M-conditions over~$P$.
		
		An M-condition over the empty graph $\emptyset$ in which all set variables are bound to quantifiers is called an $\emph{M-constraint}$.
		\qed
	\end{definition}

		For brevity, we write $\mathtt{false}$ for $\neg\mathtt{true}$, $c\Rightarrow d$ for $\neg c \vee d$, $c \Leftrightarrow d$ for $c\Rightarrow d \wedge d \Rightarrow c$, $\forall_\mathtt{V}\mathtt{X}\mt{[}c\mt{]}$ for $\neg\exists_\mathtt{V}\mathtt{X}\mt{[}\neg c\mt{]}$, $\forall_\mathtt{E}\mathtt{X}\mt{[}c\mt{]}$ for $\neg\exists_\mathtt{E}\mathtt{X}\mt{[}\neg c\mt{]}$, $\exists_\mathtt{V}\mathtt{X_1}\mt{,}\dots\mathtt{X_n}\mt{[}c\mt{]}$ for $\exists_\mathtt{V}\mathtt{X}_1\mt{[}\dots \exists_\mathtt{V}\mathtt{X_n}\mt{[}c\mt{]}\dots\mt{]}$ (analogous for other set quantifiers), $\exists (a \mid \gamma)$ for $\exists (a \mid \gamma,\mathtt{true})$, $\exists (a,c')$ for $\exists(a\mid\mathtt{true},c')$, and $\forall(a \mid \gamma, c')$ for $\neg \exists (a \mid \gamma, \neg c')$.
		
		In our examples, when the domain of a morphism $a\colon P \hookrightarrow C$ can unambiguously be inferred, we write only the codomain $C$. For instance, an M-constraint $\exists (\emptyset \hookrightarrow C, \exists (C \hookrightarrow C'))$ can be written as $\exists (C, \exists (C'))$.

	\begin{definition}[Satisfaction of M-conditions]\rm
		Let $p\!:P\hookrightarrow G$ denote an injective morphism, $c$ an M-condition over $P$, and $I$ an interpretation in $G$. We define inductively the meaning of $p \models^I c$, which denotes that $p$ \emph{satisfies} $c$ \emph{with respect to} $I$. If $c$ has the form $\mathtt{true}$, then $p \models^I c$. If $c$ has the form $\exists_\mathtt{V}\mathtt{X}\mt{[}c'\mt{]}$ (resp.\ $\exists_\mathtt{E}\mathtt{X}\mt{[}c'\mt{]}$), then $p \models^I c$ if $p \models^{I'} c'$, where $I' = I \cup \{\mathtt{X}\mapsto V\}$ for some $V \subseteq V_G$ (resp.\ $\{\mathtt{X}\mapsto E\}$ for some $E\subseteq E_G$). If $c$ has the form $\exists (a\!: P \hookrightarrow C \mid \gamma, c')$, then $p \models^I c$ if there is an injective morphism $q\!:C\hookrightarrow G$ such that $q\circ a = p$, $\gamma^{I,q} = \text{true}$, and $q \models^{I} c'$.
		
		A graph $G$ \emph{satisfies} an M-constraint $c$, denoted $G \models c$, if $i_G\!:\emptyset \hookrightarrow G \models^{I_\emptyset} c$, where $I_\emptyset$ is the \emph{empty interpretation in} $G$, i.e.\ undefined on all set variables.		
	\qed
	\end{definition}
	
	We remark that model checking for both first-order and monadic second-order logic is known to be PSPACE-complete \cite{Flum-Grohe06a}. However, the model checking problem for monadic second-order logic on graphs of bounded treewidth can be solved in linear time \cite{Courcelle90b}.

	\begin{example}\label{eg:2colouring}\rm
		The following M-constraint $col$ (translated from the corresponding formula \S 1.5 of \cite{Courcelle90a}) expresses that a graph is 2-colourable (or bipartite); i.e.\ every node can be assigned one of two colours such that no two adjacent nodes have the same one. Let $\gamma_\mathtt{col}$ denote $\mathtt{not\ }\mt{(}\mathtt{v\!\in\!X\ and\ w\!\in\!X}\mt{)}\mathtt{\ and\ not\ }\mt{(}\mathtt{v\!\in\!Y\ and\ w\!\in\!Y}\mt{)}$.\\
		
		\begin{tabular}{rll}
			\hspace{-7pt}\vspace{5pt}$\exists_\mathtt{V}\mathtt{X}\mt{,}\mathtt{Y}$ & $\mt{[}\ $&$\forall( \begin{graph}(0.5,0.3)(0,0) \fillednodes \roundnode{A}(0.2,0.075) \opaquetextfalse  \autonodetext{A}[se]{\vspace{4pt}\scriptsize $\mathtt{v}$} \end{graph},$  $\exists( \begin{graph}(0.5,0.3)(0,0) \fillednodes \roundnode{A}(0.2,0.075) \opaquetextfalse \autonodetext{A}[se]{\vspace{4pt}\scriptsize $\mathtt{v}$} \end{graph} \mid \mt{(}\mathtt{v\!\in\!X\ or\ v\!\in\!Y}\mt{)}\mathtt{\ and\ not\ }\mt{(}\mathtt{v\!\in\!X\ and\ v\!\in\!Y}\mt{)}))$ \\
			
			\vspace{5pt}&& $\wedge\ \forall( \begin{graph}(1,0.3)(0,0) \fillednodes \roundnode{A}(0.2,0.075) \opaquetextfalse  \autonodetext{A}[se]{\vspace{4pt}\scriptsize $\mathtt{v}$} \roundnode{B}(0.7,0.075) \opaquetextfalse  \autonodetext{B}[se]{\vspace{4pt}\scriptsize $\mathtt{w}$} \end{graph} , \exists(\begin{graph}(1.1,0.3)(0,0) \fillednodes  \roundnode{A}(0.2,0.075) \opaquetextfalse \autonodetext{A}[se]{\vspace{4pt}\scriptsize $\mathtt{v}$} \roundnode{B}(0.8,0.075) \opaquetextfalse \autonodetext{B}[se]{\vspace{4pt}\scriptsize $\mathtt{w}$} \diredge{A}{B} 
			 \end{graph}) \Rightarrow \exists( \begin{graph}(1,0.3)(0,0) \fillednodes \roundnode{A}(0.2,0.075) \opaquetextfalse  \autonodetext{A}[se]{\vspace{4pt}\scriptsize $\mathtt{v}$} \roundnode{B}(0.7,0.075) \opaquetextfalse  \autonodetext{B}[se]{\vspace{4pt}\scriptsize $\mathtt{w}$} \end{graph} \mid \gamma_\mathtt{col}))\ \mt{]}$ \\
			
		\end{tabular}\\
				
		\noindent A graph $G$ will satisfy $col$ if there exist two subsets of $V_G$ such that: (1) every node in $G$ belongs to \emph{exactly one} of the two sets; and (2) if there is an edge from one node to another, then those nodes are not in the same set. Intuitively, one can think of the sets $\mathtt{X}$ and $\mathtt{Y}$ as respectively denoting the nodes of colour one and colour two. If two such sets do not exist, then the graph cannot be assigned a 2-colouring.
		\qed
	\end{example}
	
	\begin{theorem}[M-constraints are equivalent to MSO formulae]\label{thm:formulae_equiv}\rm
		The assertion languages of M-constraints and MSO graph formulae are equivalently expressive: that is, given an M-constraint $c$, there exists an MSO graph formula $\varphi$ such that for all graphs $G$, $G\models c$ if and only if $G \models \varphi$; and vice versa.
		\qed
	\end{theorem}
	
	\begin{proof}
		See Appendix \ref{app:proofs:expressiveness}.
	\end{proof}

		\section{Graph Programs}\label{sec:graph_programs}

		In this section we define rules, rule application, and graph programs. Whilst the syntax and semantics of the control constructs are based on those of GP 2 \cite{Plump12a}, the rules themselves follow \cite{Habel-Pennemann-Rensink06a,Pennemann09a}, i.e.\ are labelled over a fixed finite alphabet, and do not support relabelling or expressions. We equip the rules with application conditions (M-conditions over the left- and right-hand graphs), and define \emph{rule application} via the standard double-pushout construction \cite{Ehrig-Ehrig-Prange-Taentzer06a}.
		
		\begin{definition}[Rule; direct derivation]\label{def:rule}\rm
			A \emph{plain rule} $r' = \langle L \hookleftarrow K \hookrightarrow R \rangle$ comprises two inclusions $K\hookrightarrow L$, $K\hookrightarrow R$. We call $L,R$ the left- (resp.\ right-) hand graph and $K$ the interface. An \emph{application condition} $\text{ac} = \langle \text{ac}_L,\text{ac}_R \rangle$ for $r'$ consists of two M-conditions over $L$ and $R$ respectively. A \emph{rule} $r = \langle r', \text{ac} \rangle$ is a plain rule $r'$ and an application condition $\text{ac}$ for $r'$.
			\[\begin{graph}(4,1)
			\opaquetextfalse
			\graphlinecolour{1}\graphlinewidth{.01}

			\textnode L(1,1){$L$}
			\textnode K(2,1){$K$}
			\textnode R(3,1){$R$}

			\textnode Y(1,0){$G$}
			\textnode Z(2,0){$D$}
			\textnode X(3,0){$H$}

			\freetext(1.5,0.5){$(1)$}
			\freetext(2.5,0.5){$(2)$}

			\newcommand{\arr}[3]{ \diredge{#1}{#2}[\graphlinecolour{0}]\bowtext{#1}{#2}{-.12}{#3}[] }
			\newcommand{\arrr}[3]{ \diredge{#1}{#2}[\graphlinecolour{0}]\bowtext{#1}{#2}{.12}{#3}[] }

			\arr{K}{L}{}
			\arr{K}{R}{}
			\arr{Z}{Y}{}
			\arr{Z}{X}{}
			\arr{L}{Y}{$g$}
			\arr{K}{Z}{}
			\arrr{R}{X}{$h$}
			\end{graph}\]
			
			For a plain rule $r'$ and a morphism $K\hookrightarrow D$, a \emph{direct derivation} $G\Rightarrow_{r',g,h} H$ (short.\ $G\Rightarrow_{r'} H$ or $G\Rightarrow H$) is given by the pushouts $(1)$ and $(2)$. For a rule $r = \langle r', \text{ac} \rangle$, there is a \emph{direct derivation} $G\Rightarrow_{r,g,h} H$ if $G\Rightarrow_{r',g,h} H$, $g \models^{I_\emptyset} \text{ac}_L$, and $h \models^{I_\emptyset} \text{ac}_R$. We call $g,h$ a \emph{match} (resp.\ \emph{comatch}) for $r$. Given a set of rules $\R$, we write $G\Rightarrow_\R H$ if $G\Rightarrow_{r,g,h} H$ for some $r\in \R$.
			\qed
		\end{definition}
		
		It is known that, given a (plain) rule $r$, graph $G$, and morphism $g$ as above, there exists a direct derivation if and only if $g$ satisfies the \emph{dangling condition}, i.e.\ that no node in $g(L)\setminus g(K)$ is incident to an edge in $G\setminus g(L)$. In this case, $D$ and $H$ are determined uniquely up to isomorphism, constructed from $G$ as follows: first, remove all edges in $g(L)\setminus g(K)$ obtaining $D$. Then add disjointly all nodes and edges from $R\setminus K$ retaining their labels. For $e\in E_R\setminus E_K$, $s_H(e) = s_R(e)$ if $s_R(e) \in V_R\setminus V_K$, otherwise $s_H(e) = g_V(s_R(e))$, (targets defined analogously) resulting in the graph $H$.

		We will often give rules without the interface, writing just $L\Rightarrow R$. In such cases we number nodes that correspond in $L$ and $R$, and establish the convention that $K$ comprises exactly these nodes and that $E_K = \emptyset$ (i.e.\ $K$ can be completely inferred from $L,R$). Furthermore, if the application condition of a rule is $\langle \mathtt{true},\mathtt{true} \rangle$, then we will only write the plain rule component.
		
		We consider now the syntax and semantics of graph programs, which provide a mechanism to control the application of rules to some graph provided as input.

			\begin{definition}[Graph program]\label{def:graph_program}\rm
				\emph{(Graph) programs} are defined inductively. First, every rule (resp.\ rule set) $r,\R$ and $\tt skip$ are programs. Given programs $C,P,Q$, we have that $P;Q$, $P!$, $\mathtt{if}\ C\ \mathtt{then}\ P\ \mathtt{else}\ Q$, and $\mathtt{try}\ C\ \mathtt{then}\ P\ \mathtt{else}\ Q$ are programs.
				\qed
			\end{definition}

			Graph programs are \emph{nondeterministic}, and their execution on a particular graph could result in one of several possible outcomes. That outcome could be a graph, or it could be the special state ``fail'' which occurs when a rule (set) is not \emph{applicable} to the current graph.

			An operational semantics for programs is given in Appendix \ref{app:semantics}, but the informal meaning of the constructs is as follows. Let $G$ denote an input graph. Programs $r, \R$ correspond to rule (resp.\ rule set) application, returning $H$ if there exists some $G\Rightarrow_r H$ (resp.\ $G\Rightarrow_\R H$); otherwise fail. Program $P;Q$ denotes sequential composition. Program $P!$ denotes as-long-as-possible iteration of $P$. Finally, the conditional programs execute the first or second branch depending on whether executing $C$ returns a graph or fail, with the distinction that the $\tt if$ construct does not retain any effects of $C$, whereas the $\tt try$ construct does.

	\begin{example}\label{eg:grow_tree}\rm
		Consider the program $\tt init;\ grow!$ defined by the rules:

				\begin{center}
				\fbox{%
			$\begin{array}{l c l}
			\mathtt{init:}  & \hspace{0.25in} & \mathtt{grow:}
			\\
		    \begin{graph}(.5,.5)
		     \freetext(0.5,0.2){$\emptyset$}
		    \end{graph}
			    \begin{graph}(.5,.5)
			     \freetext(0.45,0.2){$\dder$}
			    \end{graph}
		    \begin{graph}(0.5,.5)
			\fillednodes
		     \roundnode{A}(0.4,0.2)

		    \end{graph}

		&&

		    \begin{graph}(0.4,.5)
			\fillednodes
		     \roundnode{A}(0.4,0.2)
		     \autonodetext{A}[s]{\tiny 1}

		    \end{graph}
		    \begin{graph}(.5,.2)
		     \freetext(0.5,0.2){$\dder$}
		    \end{graph}
	    \begin{graph}(1.4,.5)
		\fillednodes
	      \roundnode{A}(0.4,0.2)
	      \autonodetext{A}[s]{\tiny 1}

	      \roundnode{B}(1.2,0.2)

			\diredge{A}{B}
	    \end{graph}

		\\

				&& \text{ac}_L = \neg tc

			\end{array}$

			}
		\end{center}

\noindent	where $tc$ is an (unspecified) M-condition over $L$ expressing some termination condition for the iteration (proving termination is not our concern here, see e.g.\ \cite{Poskitt-Plump13a}). The program, if executed on the empty graph, nondeterministically constructs and returns a tree. It applies the rule $\tt init$ exactly once, creating an isolated node. It then iteratively applies the rule $\tt grow$ (each application adding a leaf to the tree) until the termination condition $tc$ holds. An example program run, with $tc = \exists(\begin{graph}(0.9,0.3)(0,0) \fillednodes \roundnode{A}(0.1,0.075) \opaquetextfalse  \autonodetext{A}[se]{\vspace{4pt}\scriptsize $\mathtt{1}$} \roundnode{B}(0.4,0.075) \roundnode{C}(0.6,0.075) \roundnode{D}(0.8,0.075)  \end{graph})$, is:
\[ \emptyset \Rightarrow \begin{graph}(0.2,0.3)(0,0) \fillednodes \roundnode{A}(0.1,0.075) \end{graph} \Rightarrow \begin{graph}(0.7,0.3)(0,0) \fillednodes \roundnode{A}(0.1,0.075) \roundnode{B}(0.6,0.075) \diredge{A}{B} \end{graph} \Rightarrow \begin{graph}(0.7,0.3)(0,0) \fillednodes \roundnode{A}(0.1,0.075) \roundnode{B}(0.6,0.075) \roundnode{C}(0.45,0.3) \diredge{A}{B} \diredge{A}{C} \end{graph} \Rightarrow \begin{graph}(1.2,0.3)(0,0) \fillednodes \roundnode{A}(0.1,0.075) \roundnode{B}(0.6,0.075) \roundnode{D}(1.1,0.075) \roundnode{C}(0.45,0.3) \diredge{A}{B} \diredge{A}{C} \diredge{B}{D} \end{graph} \]
	\qed
	\end{example}

\section{Constructing a Weakest Liberal Precondition}\label{sec:constructing_wlps}
	
In this section, we present a construction for the \emph{weakest liberal precondition} relative to a rule $r$ and a postcondition $c$ (which is an M-constraint). In our terminology, if a graph satisfies a weakest liberal precondition, then: (1) any graphs resulting from applications of $r$ will satisfy $c$; and (2) there does not exist another M-constraint with this property that is weaker. (Note that we do not address termination or existence of results in this paper.)

The construction is adapted from the one for nested conditions in \cite{Habel-Pennemann09a}, and as before, is broken down into a number of stages. First, a translation of postconditions into M-conditions over $R$ (transformation ``A''); then, from M-conditions over $R$ into M-conditions over $L$ (transformation ``L''); and finally, from M-conditions over $L$ into an M-constraint expressing the weakest liberal precondition (via transformations ``App'' and ``Pre'').

First, we consider transformation A, which constructs an M-condition over $R$ from a postcondition (an M-constraint) by computing a disjunction over all the ways that the M-constraint and comatches might ``overlap''.

\begin{theorem}[M-constraints to M-conditions over $R$]\label{thm:A}\rm
	There is a transformation A, such that for all M-constraints $c$, all rules $r$ with right-hand side $R$, and all injective morphisms $h\!: R\hookrightarrow H$,
	\[ h \models^{I_\emptyset} \text{A}(r,c)\ \ \text{if and only if}\ \ H \models c. \]
\end{theorem}

\noindent	\emph{Construction.} Let $c$ denote an M-constraint, and $r$ a rule with right-hand side $R$. We define $\text{A}(r,c) = \text{A}'(\emptyset\hookrightarrow R,c)$ where $\text{A}'$ is defined inductively as follows. For injective graph morphisms $p\!: P\hookrightarrow P'$ and M-conditions over $P$, define:

	\begin{eqnarray*}
	\text{A}'(p,\mathtt{true}) &=& \mathtt{true}, \\
	\text{A}'(p,\exists_\mathtt{V}\mathtt{X}\mt{[}c'\mt{]}) &=& \exists_\mathtt{V}\mathtt{X}\mt{[}\text{A}'(p,c')\mt{]}, \\
	\text{A}'(p,\exists_\mathtt{E}\mathtt{X}\mt{[}c'\mt{]}) &=& \exists_\mathtt{E}\mathtt{X}\mt{[}\text{A}'(p,c')\mt{]}, \\
	\text{A}'(p,\exists(a\!:P\hookrightarrow C\mid\gamma,c')) &=& \textstyle{\bigvee}_{e \in \varepsilon} \exists (b\!:P'\hookrightarrow E \mid\gamma, \text{A}'(s\!:C\hookrightarrow E,c')).
	\end{eqnarray*}\\

	\noindent	The final equation relies on the following. First, construct the pushout $(1)$ of $p$ and $a$ leading to injective graph morphisms $a'\!: P' \hookrightarrow C'$ and $q\!: C \hookrightarrow C'$.

\begin{wrapfigure}[10]{r}{0.35\textwidth}
	\centering
		\vspace{-25pt}
		\includegraphics[width=0.3\textwidth]{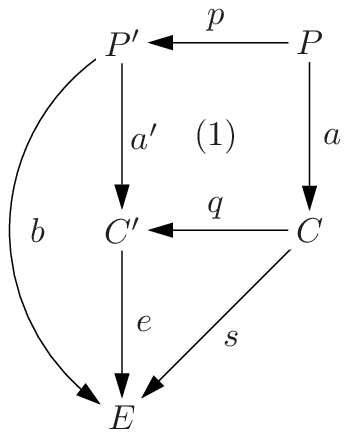}
\end{wrapfigure}
	\noindent The disjunction then ranges over the set $\varepsilon$, which we define to contain every surjective graph morphism $e\!: C' \rightarrow E$ such that $b = e \circ a'$ and $s = e \circ q$ are injective graph morphisms (we consider the codomains of each $e$ up to isomorphism, hence the disjunction is finite).

	The transformations $\text{A},\text{A}'$ are extended for Boolean formulae over M-conditions in the usual way, that is, $\text{A}(r,\neg c) = \neg\text{A}(r,c)$, $\text{A}(r,c_1\wedge c_2) = \text{A}(r,c_1)\wedge\text{A}(r,c_2)$, and $\text{A}(r,c_1\vee c_2) = \text{A}(r,c_1)\vee\text{A}(r,c_2)$ (analogous for $\text{A}'$).
	\qed

\begin{example}\label{eg:grow_A_new}\rm
	Recall the rule $\tt grow$ from Example \ref{eg:grow_tree}. Let $c$ denote the M-constraint:
	\[ \exists_\mathtt{V}\mathtt{X}\mt{,}\mathtt{Y} \mt{[}\ \forall (\begin{graph}(0.8,0.3)(0,0) \fillednodes \roundnode{A}(0.2,0.075) \opaquetextfalse  \autonodetext{A}[se]{\vspace{4pt}\scriptsize $\mathtt{v}$} \roundnode{B}(0.5,0.075) \opaquetextfalse  \autonodetext{B}[se]{\vspace{4pt}\scriptsize $\mathtt{w}$} \end{graph} , \exists(\begin{graph}(0.8,0.3)(0,0) \fillednodes \roundnode{A}(0.2,0.075) \opaquetextfalse  \autonodetext{A}[se]{\vspace{4pt}\scriptsize $\mathtt{v}$} \roundnode{B}(0.5,0.075) \opaquetextfalse  \autonodetext{B}[se]{\vspace{4pt}\scriptsize $\mathtt{w}$} \end{graph} \mid \mathtt{path}\mt{(}\mathtt{v,w}\mt{)}) \Rightarrow \exists (\begin{graph}(0.8,0.3)(0,0) \fillednodes \roundnode{A}(0.2,0.075) \opaquetextfalse  \autonodetext{A}[se]{\vspace{4pt}\scriptsize $\mathtt{v}$} \roundnode{B}(0.5,0.075) \opaquetextfalse  \autonodetext{B}[se]{\vspace{4pt}\scriptsize $\mathtt{w}$} \end{graph} \mid \gamma))    \ \mt{]} \]
	
	\noindent for $\gamma = \mt{(}\mathtt{v\in X\ and\ w\in Y}\mt{)}\ \mathtt{and\ not\ }\mt{(}\mathtt{v\in Y\ or\ w\in X}\mt{)}$, which expresses that there are two sets of nodes $X,Y$ in the graph, such that if there is a path from some node $v$ to some node $w$, then $v$ belongs only to $X$ and $w$ only to $Y$. Applying transformation A:

	\begin{center}
		\begin{tabular}{r c l}
			\multicolumn{3}{l}{$\text{A}(\mathtt{grow},c)$} \\

			&$=$& \vspace{5pt}$\text{A}'(\emptyset\hookrightarrow \begin{graph}(1.1,0.3)(0,0) \fillednodes  \roundnode{A}(0.2,0.075) \opaquetextfalse \autonodetext{A}[se]{\vspace{4pt}\scriptsize $\mathtt{1}$} \roundnode{B}(0.8,0.075) \opaquetextfalse \autonodetext{B}[se]{\vspace{4pt}\scriptsize $\mathtt{2}$} \diredge{A}{B} 
			 \end{graph},c)$\\
			
			&$=$& $\exists_\mathtt{V}\mathtt{X}\mt{,}\mathtt{Y}\mt{[}\ \text{A}'(\emptyset\hookrightarrow \begin{graph}(1.1,0.3)(0,0) \fillednodes  \roundnode{A}(0.2,0.075) \opaquetextfalse \autonodetext{A}[se]{\vspace{4pt}\scriptsize $\mathtt{1}$} \roundnode{B}(0.8,0.075) \opaquetextfalse \autonodetext{B}[se]{\vspace{4pt}\scriptsize $\mathtt{2}$} \diredge{A}{B} 
			 \end{graph},\forall (\begin{graph}(0.8,0.3)(0,0) \fillednodes \roundnode{A}(0.2,0.075) \opaquetextfalse  \autonodetext{A}[se]{\vspace{4pt}\scriptsize $\mathtt{v}$} \roundnode{B}(0.5,0.075) \opaquetextfalse  \autonodetext{B}[se]{\vspace{4pt}\scriptsize $\mathtt{w}$} \end{graph} ,$\\
			
			&& \vspace{5pt}\hspace{0.125in}$\exists(\begin{graph}(0.8,0.3)(0,0) \fillednodes \roundnode{A}(0.2,0.075) \opaquetextfalse  \autonodetext{A}[se]{\vspace{4pt}\scriptsize $\mathtt{v}$} \roundnode{B}(0.5,0.075) \opaquetextfalse  \autonodetext{B}[se]{\vspace{4pt}\scriptsize $\mathtt{w}$} \end{graph} \mid \mathtt{path}\mt{(}\mathtt{v,w}\mt{)}) \Rightarrow \exists (\begin{graph}(0.8,0.3)(0,0) \fillednodes \roundnode{A}(0.2,0.075) \opaquetextfalse  \autonodetext{A}[se]{\vspace{4pt}\scriptsize $\mathtt{v}$} \roundnode{B}(0.5,0.075) \opaquetextfalse  \autonodetext{B}[se]{\vspace{4pt}\scriptsize $\mathtt{w}$} \end{graph} \mid \gamma)))\ \mt{]}$\\

			&$=$& $\exists_\mathtt{V}\mathtt{X}\mt{,}\mathtt{Y}\mt{[}\ \bigwedge_{i=1}^7\forall (\begin{graph}(1.1,0.3)(0,0) \fillednodes  \roundnode{A}(0.2,0.075) \opaquetextfalse \autonodetext{A}[se]{\vspace{4pt}\scriptsize $\mathtt{1}$} \roundnode{B}(0.8,0.075) \opaquetextfalse \autonodetext{B}[se]{\vspace{4pt}\scriptsize $\mathtt{2}$} \diredge{A}{B} 
			 \end{graph}\hookrightarrow E_i , \exists(E_i \mid \mathtt{path}\mt{(}\mathtt{v,w}\mt{)}) \Rightarrow \exists (E_i \mid \gamma))\ \mt{]}$\\
		\end{tabular}
	\end{center}
	
	\noindent where the graphs $E_i$ are as given in Figure \ref{fig:eg-A-L}.
	\qed
\end{example}

\begin{figure}[htb]
	\centering
	\includegraphics[width=0.95\textwidth]{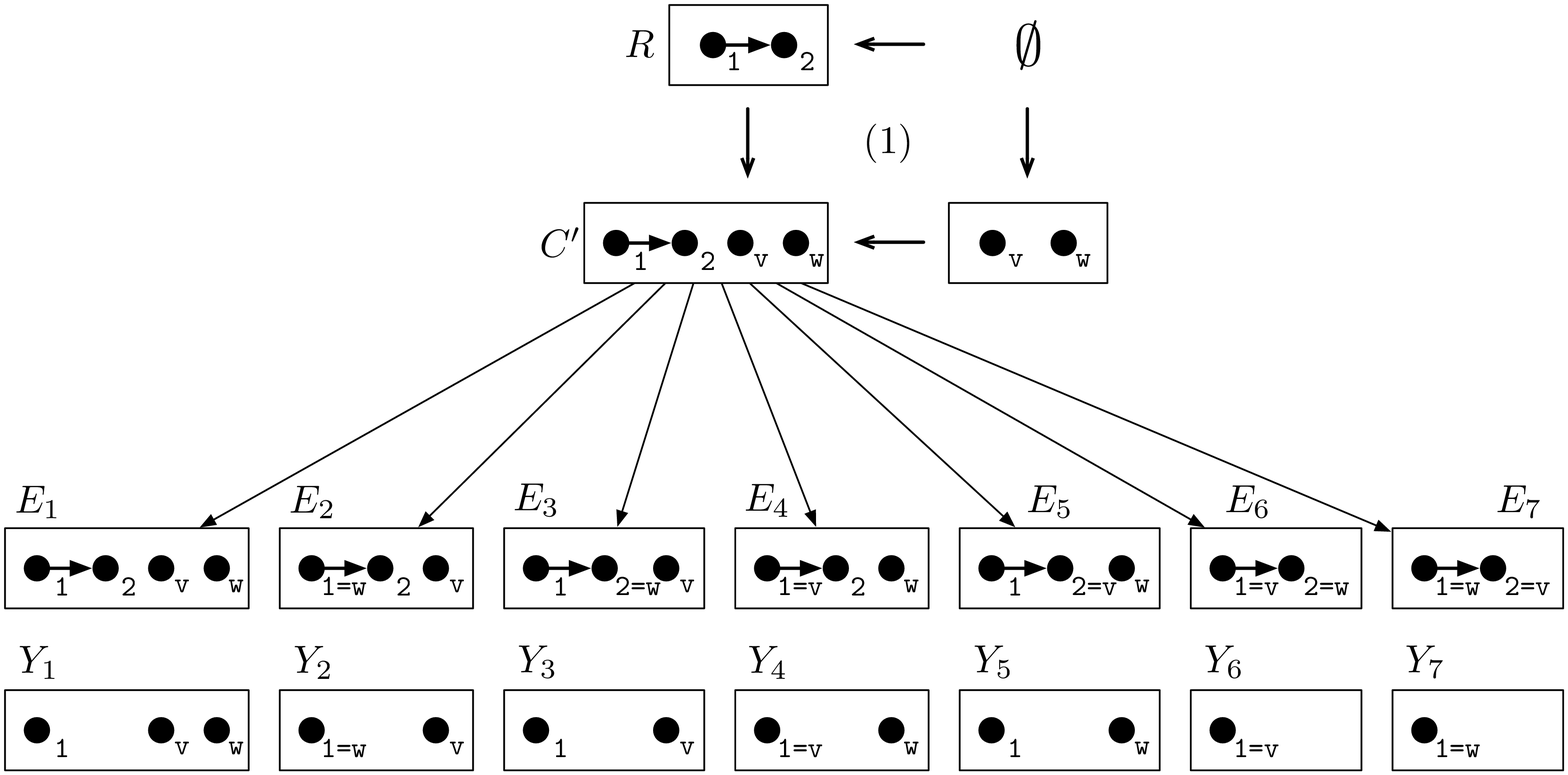}
	\caption{Applying the construction in Examples \ref{eg:grow_A_new} and \ref{eg:grow_L_new}}\label{fig:eg-A-L}
\end{figure}

In order to prove the statement of Theorem \ref{thm:A}, we first prove a more general lemma stating that an M-condition over $P$ can be shifted over a morphism $p$ with domain $P$. It is a generalised version of Lemma 3 in \cite{Habel-Pennemann09a}, but the proof is almost identical as interpretation constraints are not manipulated by this transformation, and both sides of the statement are interpreted in the same graph.

\begin{lemma}[Shifting M-conditions over morphisms]\rm\label{lemma:shifting}
	For all M-conditions $c$ over $P$, all interpretations $I$ in $H$, and all injective morphisms $p\!: P\hookrightarrow P', p''\!: P'\hookrightarrow H$, we have:
	\[ p'' \models^{I} \text{A}'(p,c)\ \ \text{if and only if}\ \ p'' \circ p \models^{I} c. \]
	\qed
\end{lemma}

\begin{proof}
	See Appendix \ref{lemma:shifting:PROOF}.
	\qed
\end{proof}

Theorem \ref{thm:A} then follows as an instance of Lemma \ref{lemma:shifting}.

\begin{proof}[of Theorem \ref{thm:A}]
	With the construction of A, Lemma \ref{lemma:shifting}, and the definition of $\models$, we have: $h \models^{I_\emptyset} \text{A}(r,c)$ iff $h \models^{I_\emptyset} \text{A}'(i_R\!:\emptyset \hookrightarrow R, c)$ iff $ h \circ i_R \models^{I_\emptyset} c$ iff $i_H\!:\emptyset\hookrightarrow H \models^{I_\emptyset} c$ iff $H \models c$.
	\qed
\end{proof}

Transformation L, adapted from \cite{Habel-Pennemann09a}, takes an M-condition over $R$ and constructs an M-condition over $L$ that is satisfied by a match if and only if the original is satisfied by the comatch. The transformation is made more complex by the presence of path and MSO expressions, because nodes and edges referred to on the right-hand side may no longer exist on the left. For clarity, we separate the handling of these two types of expressions, and in particular, define a \emph{decomposition} LPath of path predicates according to the items that the rule is creating or deleting. For example, if an edge is created by a rule, a path predicate decomposes to a disjunction of path predicates collectively asserting the existence of paths to and from the nodes that will eventually become its source and target; whereas if an edge is to be deleted, the predicate will exclude it.

\begin{proposition}[Path decomposition]\label{prop:LPath}\rm
	There is a transformation LPath such that for every rule $r = \langle L \hookleftarrow K \hookrightarrow R \rangle$, direct derivation $G\Rightarrow_{r,g,h} H$, path predicate $p$ over $R$, and interpretation $I$,
	\[ \text{LPath}(r,p)^{I,g} = p^{I,h}. \]
\end{proposition}

\noindent \emph{Construction.} Let $r = \langle L \hookleftarrow K \hookrightarrow R \rangle$ and $p = \mathtt{path}\mt{(}v,w,\mathtt{not}\ E\mt{)}$. For simplicity, we will treat the syntactic construct $E$ as a set of edges and identify $\mathtt{path}\mt{(}v,w,\mathtt{not}\ E\mt{)}$ and $\mathtt{path}\mt{(}v,w{)}$ when $E$ is empty. Then, define:
\[ \text{LPath}(r,p) = \text{LPath}'(r,v,w,E^\ominus)\ \mathtt{or}\ \text{FuturePaths}(r,p). \]

\noindent Here, $E^\ominus$ is constructed from $E$ by adding edges $e\in E_L\setminus E_R$, i.e.\ that the rule will delete. Furthermore, $\text{LPath}'(r,v,w,E^\ominus)$ decomposes to path predicates according to whether $v$ and $w$ exist in $K$. If $\text{path}_R(v,w,E^\ominus)$ holds, then $\text{LPath}'(r,v,w,E^\ominus)$ returns $\mathtt{true}$. Otherwise, if both $v,w\in V_K$, then it returns $\mathtt{path}\mt{(}v,w,\mathtt{not}\ E^\ominus\mt{)}$. If $v\notin V_K, w\in V_K$, it returns:
\[\mathtt{false\ or}\ \mathtt{path}\mt{(}x_1,w,\mathtt{not}\ E^\ominus\mt{)}\ \mathtt{or}\ \mathtt{path}\mt{(}x_2,w,\mathtt{not}\ E^\ominus\mt{)}\ \mathtt{or}\dots\]

\noindent for each $x_i\in V_K$ such that $\text{path}_R(v,x_i,E^\ominus)$. Case $v\in V_K,w\notin V_K$ analogous. If $v,w\notin V_K$, then it returns $\mathtt{false\ or}\ \mathtt{path}\mt{(}x_i,y_j,\mathtt{not}\ E^\ominus\mt{)}\ \mathtt{or}\ \dots$ for all $x_i,y_j\in V_K$ such that $\text{path}_R(v,x_i,E^\ominus)$ and $\text{path}_R(y_j,w,E^\ominus)$.

Finally, $\text{FuturePaths}(r,p)$ denotes $\tt false$ in disjunction with:
\begin{center}
	\begin{tabular}{c}
		$\mt{(}\text{LPath}'(r,v,x_1,E^\ominus)\ \mathtt{and}\ \mathtt{path}\mt{(}y_1,x_2,\mathtt{not}\ E^\ominus\mt{)} \dots \mathtt{and}\ \mathtt{path}\mt{(}y_i,x_{i+1},\mathtt{not}\ E^\ominus\mt{)}$\\
		$\dots \mathtt{and}\ \text{LPath}'(r,y_n,w,E^\ominus)\mt{)}$
	\end{tabular}
\end{center}

\noindent over all non-empty sequences of distinct pairs $\langle \langle x_1,y_1 \rangle, \dots, \langle x_n,y_n \rangle \rangle$ drawn from:
\[ \{\langle x,y \rangle \mid x,y\in V_K \wedge \text{path}_R(x,y,E^\ominus) \wedge \neg \text{path}_L(x,y,E^\ominus)\}. \]
	\qed
	
\begin{proof}
	See Section \ref{prop:LPath:PROOF}.
	
	\qed
\end{proof}

In addition to paths, transformation L must handle MSO expressions that refer to items present in $R$ but absent in $L$. To achieve this, it computes a disjunction over all possible ``future'' (i.e.\ immediately after the rule application) set memberships of these missing items. The idea being, that if a set membership exists for these missing items that satisfies the interpretation constraints \emph{before} the rule application, then one will still exist once they have been created. The transformation keeps track of such potential memberships via sets of pairs as follows.

\begin{definition}[Membership set]\rm
	A \emph{membership set} $M$ is a set of pairs $(x,\mathtt{X})$ of node or edge identifiers $x$ with set variables of the corresponding type. Intuitively, $(x,\mathtt{X}) \in M$ encodes that $x\in\mathtt{X}$, whereas $(x,\mathtt{X}) \notin M$ encodes that $x\notin\mathtt{X}$.
	\qed
\end{definition}

\begin{theorem}[From M-conditions over $R$ to $L$]\label{thm:L}\rm
	There is a transformation L such that for every rule $r = \langle \langle L \hookleftarrow K \hookrightarrow R \rangle, \text{ac} \rangle$, every M-condition $c$ over $R$ (with no free variables, and distinct variables for distinct quantifiers), and every direct derivation $G\Rightarrow_{r,g,h} H$,
	\[ g \models^{I_\emptyset} \text{L}(r,c)\ \ \text{if and only if}\ \ h \models^{I_\emptyset} c. \]
\end{theorem}
	
	\noindent \emph{Construction.} Let $r = \langle \langle L \hookleftarrow K \hookrightarrow R \rangle, \text{ac} \rangle$ denote a rule and $c$ an M-condition over $R$. We define $\text{L}(r,c) = \text{L}'(r,c,\emptyset)$. For such an $r,c$, and membership set $M$, the transformation $\text{L}'$ is defined inductively as follows:

\begin{eqnarray*}
\text{L}'(r,\mathtt{true}, M) &=& \mathtt{true}, \\
\text{L}'(r,\exists_\mathtt{V}\mathtt{X}\mt{[}c'\mt{]}, M ) &=& \exists_\mathtt{V}\mathtt{X}\mt{[}\ \bigvee_{M'\in 2^{M_\mathtt{V}}}\text{L}'(r,c', M\cup M')\ \mt{]}\\
\text{L}'(r,\exists_\mathtt{E}\mathtt{X}\mt{[}c'\mt{]}, M ) &=& \exists_\mathtt{E}\mathtt{X}\mt{[}\ \bigvee_{M'\in 2^{M_\mathtt{E}}}\text{L}'(r,c', M\cup M')\ \mt{]}
\end{eqnarray*}

	\noindent where $M_\mathtt{V} = \{ (v,\mathtt{X}) \mid v \in V_R\setminus V_L \}$ and $M_\mathtt{E} = \{ (e,\mathtt{X}) \mid e \in E_R\setminus E_L \}$.
	
	For case $c = \exists(a\mid\gamma,c')$, we define:
	\begin{eqnarray*}
		\text{L}'(r,\exists(a\mid\gamma,c'),M) &=& \mathtt{false}
	\end{eqnarray*}	
	
	\noindent if $\langle K \hookrightarrow R, a \rangle$ has no pushout complement; otherwise:

	\[ \text{L}'(r,\exists(a\mid\gamma,c'),M) = \exists(b\mid\gamma_{M},\text{L}'(r^{*},c',M)) \]

	\begin{wrapfigure}[6]{r}{0.5\textwidth}
		\centering
			\vspace{-2.5pt}
			\includegraphics[width=0.45\textwidth]{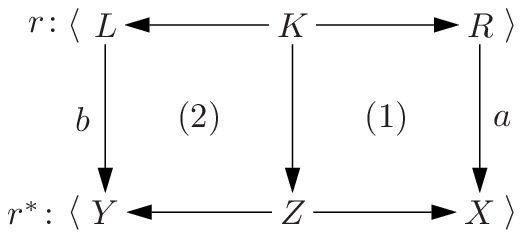}
	\end{wrapfigure}
	\noindent which relies on the following. First, construct the pushout $(1)$, with $r^{*} = \langle Y \hookleftarrow Z \hookrightarrow X \rangle$ the ``derived'' rule obtained by constructing pushout $(2)$. The interpretation constraint $\gamma_{M}$ is obtained from $\gamma$ as follows. First, consider each predicate $x\in\mathtt{X}$ such that $x\notin Y$. If $(y,\mathtt{X})\in M$ for some $y=x$, replace the predicate with $\mathtt{true}$; otherwise $\mathtt{false}$. Then, replace each path predicate $p$ with $\text{LPath}(r^*,p)$.
	
	The transformation L is extended for Boolean formulae in the usual way, that is, $\text{L}(r,\neg c) = \neg \text{L}(r,c)$, $\text{L}(r,c_1\wedge c_2) = \text{L}(r,c_1) \wedge \text{L}(r,c_2)$, and $\text{L}(r,c_1\vee c_2) = \text{L}(r,c_1) \vee \text{L}(r,c_2)$ (analogous for $\text{L}'$).	
	\qed

\begin{example}\label{eg:grow_L_new}\rm
	Take $\tt grow$, $c$, $\gamma$ and $\text{A}(\mathtt{grow},c)$ as considered in Example \ref{eg:grow_A_new}. Applying transformation L: 
	
	\begin{center}
		\begin{tabular}{r c l}
			\multicolumn{3}{l}{$\text{L}(\mathtt{grow},\text{A}(\mathtt{grow},c)) = \text{L}'(\mathtt{grow},\text{A}(\mathtt{grow},c),\emptyset)$} \\

			&$=$& \vspace{5pt}$\exists_\mathtt{V}\mathtt{X}\mt{,}\mathtt{Y}\mt{[}\ \bigvee_{M'\in 2^{M_\mathtt{V}}}\text{L}'(\mathtt{grow},\bigwedge_{i=1}^7\forall (\begin{graph}(1.1,0.3)(0,0) \fillednodes  \roundnode{A}(0.2,0.075) \opaquetextfalse \autonodetext{A}[se]{\vspace{4pt}\scriptsize $\mathtt{1}$} \roundnode{B}(0.8,0.075) \opaquetextfalse \autonodetext{B}[se]{\vspace{4pt}\scriptsize $\mathtt{2}$} \diredge{A}{B} 
			 \end{graph}\hookrightarrow E_i , \exists(E_i \mid \mathtt{path}\mt{(}\mathtt{v,w}\mt{)})$\\
			
			&& \vspace{5pt}\hspace{0.25in}$\Rightarrow \exists (E_i \mid \gamma)),M')\ \mt{]}$\\
			
			&$=$& $\exists_\mathtt{V}\mathtt{X}\mt{,}\mathtt{Y}\mt{[}\ \bigvee_{M'\in 2^{M_\mathtt{V}}} (\bigwedge_{i\in\{1,2,4\}}\forall (\begin{graph}(0.45,0.3)(0,0) \fillednodes \roundnode{A}(0.2,0.075) \opaquetextfalse \autonodetext{A}[se]{\vspace{4pt}\scriptsize $\mathtt{1}$}   \end{graph}\hookrightarrow Y_i , \exists(Y_i \mid \mathtt{path}\mt{(}\mathtt{v,w}\mt{)}) \Rightarrow \exists (Y_i \mid \gamma))$\\

			&& $\wedge\ \forall (\begin{graph}(0.8,0.3)(0,0) \fillednodes \roundnode{A}(0.2,0.075) \opaquetextfalse \autonodetext{A}[se]{\vspace{4pt}\scriptsize $\mathtt{1}$}  \roundnode{C}(0.5,0.075) \opaquetextfalse  \autonodetext{C}[se]{\vspace{4pt}\scriptsize $\mathtt{v}$} \end{graph} , \exists(\begin{graph}(0.8,0.3)(0,0) \fillednodes \roundnode{A}(0.2,0.075) \opaquetextfalse \autonodetext{A}[se]{\vspace{4pt}\scriptsize $\mathtt{1}$}  \roundnode{C}(0.5,0.075) \opaquetextfalse  \autonodetext{C}[se]{\vspace{4pt}\scriptsize $\mathtt{v}$} \end{graph} \mid \mathtt{path}\mt{(}\mathtt{v,1}\mt{)}) \Rightarrow \exists (\begin{graph}(0.8,0.3)(0,0) \fillednodes \roundnode{A}(0.2,0.075) \opaquetextfalse \autonodetext{A}[se]{\vspace{4pt}\scriptsize $\mathtt{1}$}  \roundnode{C}(0.5,0.075) \opaquetextfalse  \autonodetext{C}[se]{\vspace{4pt}\scriptsize $\mathtt{v}$} \end{graph} \mid \gamma_{M'},\text{L}'(\mathtt{grow},\mathtt{true},M')))$\\
			
			&& $\wedge\ \forall (\begin{graph}(0.8,0.3)(0,0) \fillednodes \roundnode{A}(0.2,0.075) \opaquetextfalse \autonodetext{A}[se]{\vspace{4pt}\scriptsize $\mathtt{1}$}  \roundnode{C}(0.5,0.075) \opaquetextfalse  \autonodetext{C}[se]{\vspace{4pt}\scriptsize $\mathtt{w}$} \end{graph} , \mathtt{false} \Rightarrow \exists (\begin{graph}(0.8,0.3)(0,0) \fillednodes \roundnode{A}(0.2,0.075) \opaquetextfalse \autonodetext{A}[se]{\vspace{4pt}\scriptsize $\mathtt{1}$}  \roundnode{C}(0.5,0.075) \opaquetextfalse  \autonodetext{C}[se]{\vspace{4pt}\scriptsize $\mathtt{w}$} \end{graph} \mid \gamma_{M'},\text{L}'(\mathtt{grow},\mathtt{true},M')))$\\
			
			&& $\wedge\ \forall (\begin{graph}(0.65,0.3)(0,0) \fillednodes \roundnode{A}(0.2,0.075) \opaquetextfalse \autonodetext{A}[se]{\vspace{4pt}\scriptsize $\mathtt{1}\mt{=}\mathtt{v}$}   \end{graph} , \mathtt{true} \Rightarrow \exists (\begin{graph}(0.65,0.3)(0,0) \fillednodes \roundnode{A}(0.2,0.075) \opaquetextfalse \autonodetext{A}[se]{\vspace{4pt}\scriptsize $\mathtt{1}\mt{=}\mathtt{v}$}   \end{graph} \mid \gamma_{M'},\text{L}'(\mathtt{grow},\mathtt{true},M')))$\\
			
			&& \vspace{5pt}$\wedge\ \forall (\begin{graph}(0.65,0.3)(0,0) \fillednodes \roundnode{A}(0.2,0.075) \opaquetextfalse \autonodetext{A}[se]{\vspace{4pt}\scriptsize $\mathtt{1}\mt{=}\mathtt{w}$}   \end{graph} , \mathtt{false} \Rightarrow \exists (\begin{graph}(0.65,0.3)(0,0) \fillednodes \roundnode{A}(0.2,0.075) \opaquetextfalse \autonodetext{A}[se]{\vspace{4pt}\scriptsize $\mathtt{1}\mt{=}\mathtt{w}$}   \end{graph} \mid \gamma_{M'},\text{L}'(\mathtt{grow},\mathtt{true},M'))))\ \mt{]}$\\

			&$=$& $\exists_\mathtt{V}\mathtt{X}\mt{,}\mathtt{Y}\mt{[}\ \bigvee_{M'\in 2^{M_\mathtt{V}}} (\bigwedge_{i\in\{1,2,4\}}\forall (\begin{graph}(0.45,0.3)(0,0) \fillednodes \roundnode{A}(0.2,0.075) \opaquetextfalse \autonodetext{A}[se]{\vspace{4pt}\scriptsize $\mathtt{1}$}   \end{graph}\hookrightarrow Y_i , \exists(Y_i \mid \mathtt{path}\mt{(}\mathtt{v,w}\mt{)}) \Rightarrow \exists (Y_i \mid \gamma))$\\

			&& $\wedge\ \forall (\begin{graph}(0.8,0.3)(0,0) \fillednodes \roundnode{A}(0.2,0.075) \opaquetextfalse \autonodetext{A}[se]{\vspace{4pt}\scriptsize $\mathtt{1}$}  \roundnode{C}(0.5,0.075) \opaquetextfalse  \autonodetext{C}[se]{\vspace{4pt}\scriptsize $\mathtt{v}$} \end{graph} , \exists(\begin{graph}(0.8,0.3)(0,0) \fillednodes \roundnode{A}(0.2,0.075) \opaquetextfalse \autonodetext{A}[se]{\vspace{4pt}\scriptsize $\mathtt{1}$}  \roundnode{C}(0.5,0.075) \opaquetextfalse  \autonodetext{C}[se]{\vspace{4pt}\scriptsize $\mathtt{v}$} \end{graph} \mid \mathtt{path}\mt{(}\mathtt{v,1}\mt{)}) \Rightarrow \exists (\begin{graph}(0.8,0.3)(0,0) \fillednodes \roundnode{A}(0.2,0.075) \opaquetextfalse \autonodetext{A}[se]{\vspace{4pt}\scriptsize $\mathtt{1}$}  \roundnode{C}(0.5,0.075) \opaquetextfalse  \autonodetext{C}[se]{\vspace{4pt}\scriptsize $\mathtt{v}$} \end{graph} \mid \gamma_{M'}))$\\
					
			&& \vspace{5pt}$\wedge\ \forall (\begin{graph}(0.65,0.3)(0,0) \fillednodes \roundnode{A}(0.2,0.075) \opaquetextfalse \autonodetext{A}[se]{\vspace{4pt}\scriptsize $\mathtt{1}\mt{=}\mathtt{v}$}   \end{graph} , \exists (\begin{graph}(0.65,0.3)(0,0) \fillednodes \roundnode{A}(0.2,0.075) \opaquetextfalse \autonodetext{A}[se]{\vspace{4pt}\scriptsize $\mathtt{1}\mt{=}\mathtt{v}$}   \end{graph} \mid \gamma_{M'})))\ \mt{]}$\\

			&$=$& $\exists_\mathtt{V}\mathtt{X}\mt{,}\mathtt{Y}\mt{[}\ \bigwedge_{i\in\{1,2,4\}}\forall (\begin{graph}(0.45,0.3)(0,0) \fillednodes \roundnode{A}(0.2,0.075) \opaquetextfalse \autonodetext{A}[se]{\vspace{4pt}\scriptsize $\mathtt{1}$}   \end{graph}\hookrightarrow Y_i , \exists(Y_i \mid \mathtt{path}\mt{(}\mathtt{v,w}\mt{)}) \Rightarrow \exists (Y_i \mid \gamma))$\\
			
			&& $\wedge\ \forall (\begin{graph}(0.8,0.3)(0,0) \fillednodes \roundnode{A}(0.2,0.075) \opaquetextfalse \autonodetext{A}[se]{\vspace{4pt}\scriptsize $\mathtt{1}$}  \roundnode{C}(0.5,0.075) \opaquetextfalse  \autonodetext{C}[se]{\vspace{4pt}\scriptsize $\mathtt{v}$} \end{graph} , \exists(\begin{graph}(0.8,0.3)(0,0) \fillednodes \roundnode{A}(0.2,0.075) \opaquetextfalse \autonodetext{A}[se]{\vspace{4pt}\scriptsize $\mathtt{1}$}  \roundnode{C}(0.5,0.075) \opaquetextfalse  \autonodetext{C}[se]{\vspace{4pt}\scriptsize $\mathtt{v}$} \end{graph} \mid \mathtt{path}\mt{(}\mathtt{v,1}\mt{)}) \Rightarrow \exists (\begin{graph}(0.8,0.3)(0,0) \fillednodes \roundnode{A}(0.2,0.075) \opaquetextfalse \autonodetext{A}[se]{\vspace{4pt}\scriptsize $\mathtt{1}$}  \roundnode{C}(0.5,0.075) \opaquetextfalse  \autonodetext{C}[se]{\vspace{4pt}\scriptsize $\mathtt{v}$} \end{graph} \mid \mathtt{v\in X\ and\ not\ v\in Y}))$\\
			
			&& $\wedge\ \forall (\begin{graph}(0.65,0.3)(0,0) \fillednodes \roundnode{A}(0.2,0.075) \opaquetextfalse \autonodetext{A}[se]{\vspace{4pt}\scriptsize $\mathtt{1}\mt{=}\mathtt{v}$}   \end{graph} , \exists (\begin{graph}(0.65,0.3)(0,0) \fillednodes \roundnode{A}(0.2,0.075) \opaquetextfalse \autonodetext{A}[se]{\vspace{4pt}\scriptsize $\mathtt{1}\mt{=}\mathtt{v}$}   \end{graph} \mid \mathtt{v\in X\ and\ not\ v\in Y}))\ \mt{]}$\\
			
		\end{tabular}
	\end{center}
	
	\noindent where the graphs $E_i$ and $Y_i$ are as given in Figure \ref{fig:eg-A-L} and $M_\mathtt{V} = \{(2,\mathtt{X}),(2,\mathtt{Y})\}$. Here, only one of the subsets ranged over yields a satisfiable disjunct: $M' = \{(2,\mathtt{Y})\}$, i.e. $\gamma_{M'} = \mt{(}\mathtt{v\in X\ and\ true}\mt{)}\ \mathtt{and\ not\ }\mt{(}\mathtt{v\in Y\ or\ false}\mt{)}$ for $\mathtt{w}=2$.
	\qed
\end{example}

In order to prove the statement about L (which is interpreted over $I_\emptyset$), we need to prove a more general lemma.

\begin{lemma}\label{lemma:L}\rm
	There is a transformation L such that for every rule $r = \langle \langle L \hookleftarrow K \hookrightarrow R \rangle, \text{ac} \rangle$, every M-condition $c$ over $R$ with distinct variables for distinct quantifiers, every interpretation $I$ in $G$ defined for all free set variables of $c$, every membership set $M$ such that $(x,\_)\in M$ implies $x\in R\setminus L$, and every direct derivation $G\Rightarrow_{r,g,h} H$,
	\[ g \models^{I} \text{L}'(r,c,M)\ \ \text{if and only if}\ \ h \models^{I_M} c. \]
	
	\noindent Here, $I_M$ is defined as $I$ except for all $x\in R\setminus L$, where $h(x) \in I_M(\mathtt{X})$ if and only if $(x,\mathtt{X})\in M$.
	\qed
\end{lemma}

\begin{proof}
	See Appendix \ref{lemma:L:PROOF}.
	\qed
\end{proof}

\begin{proof}[of Theorem \ref{thm:A}]\rm
	With the construction of L, Lemma \ref{lemma:L}, and the definition of $\models$, we have: $g \models^{I_\emptyset} \text{L}(r,c)$ iff $g \models^{I_\emptyset} \text{L}'(r,c,\emptyset)$ iff $h \models^{I_\emptyset} c$.
	\qed
\end{proof}

Transformation App, adapted from Def in \cite{Pennemann09a}, takes as input a rule set $\R$ and generates an M-constraint that is satisfied by graphs for which $\R$ is applicable.

\begin{theorem}[Applicability of a rule]\label{thm:A}\rm
	There is a transformation App such that for every rule set $\R$ and every graph $G$,
	\[ G \models \text{App}(\R)\ \ \text{if and only if}\ \ \exists H.\ G \Rightarrow_{\R} H. \]

\noindent \emph{Construction.} If $\R$ is empty, define $\text{App}(\R)=\mathtt{false}$; otherwise, for $\R = \{r_1,\dots,r_n\}$, define:
	\[ \text{App}(\R) = \text{app}(r_1) \vee \dots \vee \text{app}(r_n). \]
	
	\noindent For each rule $r = \langle r',\text{ac}\rangle$ with $r' = \langle L \hookleftarrow K \hookrightarrow R \rangle$, we define $\text{app}(r) = \exists(\emptyset\hookrightarrow L,\text{Dang}(r') \wedge \text{ac}_L \wedge \text{L}(r,\text{ac}_R))$. Here, $\text{Dang}(r') = \bigwedge_{a\in A} \neg \exists a$, where the index set $A$ ranges over all injective graph morphisms $a\!:L \hookrightarrow L^\oplus$ (up to isomorphic codomains) such that the pair $\langle K \hookrightarrow L, a \rangle$ has no pushout complement; each $L^\oplus$ a graph that can be obtained from $L$ by adding either (1) a loop; (2) a single edge between distinct nodes; or (3) a single node and a non-looping edge incident to that node.
	\qed
\end{theorem}

\begin{proof}
	See the corresponding proofs in \cite{Pennemann09a} and \cite{Poskitt13a} for nested conditions and E-conditions respectively. (The difference is in the application conditions, i.e.\ M-conditions over $L$ and $R$. Correctness follows from the definition of $\models$ for M-conditions and Theorem \ref{thm:L}.)
	\qed
\end{proof}

Finally, transformation Pre (adapted from \cite{Habel-Pennemann-Rensink06a}) combines the other transformations to construct a weakest liberal precondition relative to a rule and postcondition.

\begin{theorem}[Postconditions to weakest liberal preconditions]\label{thm:Pre}\rm
	There is a transformation Pre such that for every rule $r = \langle \langle L \hookleftarrow K \hookrightarrow R \rangle, \text{ac} \rangle$, every M-constraint $c$, and every direct derivation $G\Rightarrow_r H$,
	\[ G \models \text{Pre}(r,c)\ \ \text{if and only if}\ \ H \models c. \]
	
	\noindent Moreover, $\text{Pre}(r,c) \vee \neg \text{App}(\{r\})$ is the \emph{weakest liberal precondition} relative to $r$ and $c$.\\

\noindent \emph{Construction.} Let $r = \langle \langle L \hookleftarrow K \hookrightarrow R \rangle, \text{ac} \rangle$ denote a rule and $c$ denote an M-constraint. Then:
	\[ \text{Pre}(r,c) = \forall(\emptyset\hookrightarrow L, (\text{Dang}(r) \wedge \text{ac}_L \wedge \text{L}(r,\text{ac}_R)) \Rightarrow \text{L}(r,\text{A}(r,c))). \]
	\qed
\end{theorem}
	
\begin{proof}
	As for nested conditions (see \cite{Pennemann09a}), but adapted for the definition of $\models$ for M-conditions and Theorem \ref{thm:L}. That Pre and App can be used to construct the weakest liberal precondition is shown in \cite{Poskitt13a}.
	\qed
\end{proof}
	
\begin{example}\label{eg:grow_Pre_new}\rm
	Take $\tt grow$, $c$, $\gamma$ and $\text{L}(\mathtt{grow},\text{A}(\mathtt{grow},c))$ as considered in Example \ref{eg:grow_L_new}. Applying transformation Pre:
	
	\begin{center}
		\begin{tabular}{r c l}
			\multicolumn{3}{l}{$\text{Pre}(\mathtt{grow},\text{L}(\mathtt{grow},\text{A}(\mathtt{grow},c)))$} \\

			&$=$& $\forall (\begin{graph}(0.45,0.3)(0,0) \fillednodes \roundnode{A}(0.2,0.075) \opaquetextfalse \autonodetext{A}[se]{\vspace{4pt}\scriptsize $\mathtt{1}$}   \end{graph}, \text{ac}_L\Rightarrow \exists_\mathtt{V}\mathtt{X}\mt{,}\mathtt{Y}\mt{[}\ \bigwedge_{i\in\{1,2,4\}}\forall (\begin{graph}(0.5,0.3)(0,0) \fillednodes \roundnode{A}(0.2,0.075) \opaquetextfalse \autonodetext{A}[se]{\vspace{4pt}\scriptsize $\mathtt{1}$}   \end{graph}\hookrightarrow Y_i , \exists(Y_i \mid \mathtt{path}\mt{(}\mathtt{v,w}\mt{)}) \Rightarrow \exists (Y_i \mid \gamma))$\\
			
			&& \hspace{0.125in}$\wedge\ \forall (\begin{graph}(0.8,0.3)(0,0) \fillednodes \roundnode{A}(0.2,0.075) \opaquetextfalse \autonodetext{A}[se]{\vspace{4pt}\scriptsize $\mathtt{1}$}  \roundnode{C}(0.5,0.075) \opaquetextfalse  \autonodetext{C}[se]{\vspace{4pt}\scriptsize $\mathtt{v}$} \end{graph} , \exists(\begin{graph}(0.8,0.3)(0,0) \fillednodes \roundnode{A}(0.2,0.075) \opaquetextfalse \autonodetext{A}[se]{\vspace{4pt}\scriptsize $\mathtt{1}$}  \roundnode{C}(0.5,0.075) \opaquetextfalse  \autonodetext{C}[se]{\vspace{4pt}\scriptsize $\mathtt{v}$} \end{graph} \mid \mathtt{path},\mt{(}\mathtt{v,1}\mt{)}) \Rightarrow \exists (\begin{graph}(0.8,0.3)(0,0) \fillednodes \roundnode{A}(0.2,0.075) \opaquetextfalse \autonodetext{A}[se]{\vspace{4pt}\scriptsize $\mathtt{1}$}  \roundnode{C}(0.5,0.075) \opaquetextfalse  \autonodetext{C}[se]{\vspace{4pt}\scriptsize $\mathtt{v}$} \end{graph} \mid \mathtt{v\in X\ and\ not\ v\in Y}))$\\

			&& \hspace{0.125in}$\wedge\ \forall (\begin{graph}(0.65,0.3)(0,0) \fillednodes \roundnode{A}(0.2,0.075) \opaquetextfalse \autonodetext{A}[se]{\vspace{4pt}\scriptsize $\mathtt{1}\mt{=}\mathtt{v}$}   \end{graph} , \exists (\begin{graph}(0.65,0.3)(0,0) \fillednodes \roundnode{A}(0.2,0.075) \opaquetextfalse \autonodetext{A}[se]{\vspace{4pt}\scriptsize $\mathtt{1}\mt{=}\mathtt{v}$}   \end{graph} \mid \mathtt{v\in X\ and\ not\ v\in Y}))\ \mt{]})$\\
		\end{tabular}
	\end{center}
	
	\noindent where the graphs $Y_i$ are as given in Figure \ref{fig:eg-A-L}. This M-constraint is only satisfied by graphs that do not have any edges between distinct nodes, because of the assertion that every match (i.e.\ every node) must be in $\mathtt{X}$ and not in $\mathtt{Y}$. Were an edge to exist -- i.e.\ a path -- then the M-constraint asserts that its target is in $\mathtt{Y}$; a contradiction.
	\qed
\end{example}

	\section{Proving Non-Local Specifications}\label{sec:example_proofs}
	
	In this section we show how to systematically prove a non-local correctness specification using a Hoare logic adapted from \cite{Poskitt-Plump12a,Poskitt13a}. The key difference is the use of M-constraints as assertions, and our extension of Pre in constructing weakest liberal preconditions for rules. (We note that one could just as easily adapt the Dijkstra-style systems of \cite{Habel-Pennemann-Rensink06a,Pennemann09a}.)
	
	We will specify the behaviour of programs using \emph{(Hoare) triples}, $\{c\}\ P\ \{d\}$, where $P$ is a program, and $c,d$ are \emph{pre-} and \emph{postconditions} expressed as M-constraints. We say that this specification holds in the sense of \emph{partial correctness}, denoted by $\models \{c\}\ P\ \{d\}$, if for any graph $G$ satisfying $c$, every graph $H$ resulting from the execution of $P$ on $G$ satisfies $d$.
	
	For systematically proving a specification, we present a \emph{Hoare logic} in Figure \ref{fig:a_hoare_logic}, where $c,d,e,inv$ range over M-constraints, $P,Q$ over programs, $r$ over rules, and $\R$ over rule sets. If a triple $\{c\}\ P\ \{d\}$ can be instantiated from an axiom or deduced from an inference rule, then it is \emph{provable} in the Hoare logic and we write $\vdash \{c\}\ P\ \{d\}$. Proofs shall be displayed as trees, with the specification as the root, axiom instances as the leaves, and inference rule instances in-between.
	
	\begin{figure}[htb]
	\vspace{-25pt}
	{\footnotesize\begin{center}
	\begin{tabular}{ p{0.5\textwidth} p{0.5\textwidth} }

	\vspace{2pt}\begin{prooftree}
	\AxiomC{[ruleapp]$_\text{wlp}$\ \ $\{ \text{Pre}(r,c) \vee \neg\text{App}(\{r\}) \}~ r ~ \{ c \}$}
	\end{prooftree}
	&
	\begin{prooftree}
	\AxiomC{$\{ c \}~ r~ \{ d \}~ \text{for each $r\in\R$}$}
	\LeftLabel{[$\text{ruleset}$]}
	\UnaryInfC{$\{ c \}~ \R~ \{ d \}$}
	\end{prooftree}
	\\[-15pt]

	\begin{prooftree}
	\AxiomC{$\{ c \}~ P~ \{ e \}$}
	\AxiomC{$\{ e \}~ Q~ \{ d \}$}
	\LeftLabel{[comp]}
	\BinaryInfC{$\{ c \}~ P\mathtt{;}~ Q ~ \{ d \}$}
	\end{prooftree}

	&
	\begin{prooftree}
	\AxiomC{$\{ inv \}~ \mathcal{R}~ \{ inv \}$}
	\LeftLabel{[!]}
	\UnaryInfC{$\{ inv \}~ \mathcal{R}\mathtt{!} ~ \{ inv \wedge \neg\text{App}(\mathcal{R}) \}$}
	\end{prooftree}
	\\[-15pt]

	\multicolumn{2}{p{\textwidth}}{

	\begin{prooftree}
	\AxiomC{$c \Rightarrow c'\ \ \{ c' \}~ P~ \{ d' \}\ \ d' \Rightarrow d$}
	\LeftLabel{[cons]} 
	\UnaryInfC{$\{ c \}~ P ~ \{ d \}$}
	\end{prooftree}
	
	}\\
	\end{tabular}\end{center}\vspace*{-4mm}
	}
	\caption{A Hoare logic for partial correctness}\label{fig:a_hoare_logic}
	\end{figure}
	
	 For simplicity in proofs we will typically treat [ruleapp]$_\text{wlp}$ as two different axioms (one for each disjunct). Note that we have omitted, due to space, the proof rules for the conditional constructs. Note also the restriction to rule sets in [!], because the applicability of arbitrary programs cannot be expressed in a logic for which the model checking problem is decidable \cite{Poskitt13a}.
	
	\begin{theorem}[Soundness]\rm
		Given a program $P$ and M-constraints $c,d$, we have that $\vdash \{c\}\ P\ \{d\}$ implies $\models \{c\}\ P\ \{d\}$.
		\qed
	\end{theorem}
	
\begin{proof}
	See \cite{Poskitt13a} for a soundness proof of the corresponding extensional partial correctness calculus.
	\qed
\end{proof}

	The remainder of this section demonstrates the use of our constructions and Hoare logic in proving non-local specifications of two programs. For the first, we will consider a property expressed in terms of MSO variables and expressions, whereas for the second, we will consider properties expressed in terms of $\tt path$ predicates. Both programs are simple, as our focus here is not on building intricate proofs but rather on illustrating the main novelty of this paper: a Pre construction for MSO properties.
	
	\begin{example}\rm
	Recall the program $\mathtt{init;\ grow!}$ of Example \ref{eg:grow_tree} that nondeterministically constructs a tree. A known non-local property of trees is that they can be assigned a 2-colouring (i.e.\ they are bipartite), a property that the M-constraint $col$ of Example \ref{eg:2colouring} precisely expresses. Hence we will show that $\vdash \{emp\}\ \mathtt{init;\ grow!}\ \{col\}$, where $emp = \neg \exists(\begin{graph}(0.4,0.3)(0,0) \fillednodes \roundnode{A}(0.2,0.075)   \end{graph})$ expresses that the graph is empty. A proof tree for this specification is given in Figure \ref{eg:trees_2col}, where the interpretation constraints $\gamma_1$ and $\gamma_2$ in $\text{Pre}(\mathtt{grow},col)$ are respectively $\mt{(}\mathtt{v\!\in\!X\ or\ v\!\in\!Y}\mt{)}$ $\mathtt{and\ not\ }\mt{(}\mathtt{v\!\in\!X\ and\ v\!\in\!Y}\mt{)}$ and $\mathtt{not\ }\mt{(}\mathtt{v\!\in\!X\ and\ w\!\in\!X}\mt{)}\mathtt{\ and\ not\ }\mt{(}\mathtt{v\!\in\!Y\ and\ w\!\in\!Y}\mt{)}$.

\begin{figure}[htb]
	\centering	
	{\small
		\begin{prooftree}

\AxiomC{$\{\text{Pre}(\mathtt{init},col)\}\ \mathtt{init}\ \{col\}$}

	\UnaryInfC{$\{ emp  \}~ \mathtt{init}~ \{ col \}$}

\AxiomC{$\{\text{Pre}(\mathtt{grow},col)\}\ \mathtt{grow}\ \{col\} $}

	\UnaryInfC{$\{ col \}~ \mathtt{grow}~ \{ col \}$}

	\UnaryInfC{$\{ col \}~ \mathtt{grow}!~ \{ col \wedge \neg\text{App}(\mathtt{\{grow\}}) \}$}

	\BinaryInfC{$\vdash \{emp\}\ \mathtt{init; grow!}\ \{col\}$}
	\end{prooftree}}

		\begin{tabular}{r c l}
			$\text{Pre}(\mathtt{init},col)$ &$\equiv$& $col$ \\

			$\text{Pre}(\mathtt{grow},col)$ &$\equiv$& $\forall(\begin{graph}(0.5,0.3)(0,0) \fillednodes \roundnode{A}(0.2,0.075) \opaquetextfalse \autonodetext{A}[se]{\vspace{4pt}\scriptsize $\mathtt{1}$}  \end{graph}, \neg tc \Rightarrow \exists_\mathtt{V}\mathtt{X}\mt{,}\mathtt{Y}\mt{[}\  $\\
			
			&& \hspace{0.25in}$\forall(\begin{graph}(0.8,0.3)(0,0) \fillednodes \roundnode{A}(0.2,0.075) \opaquetextfalse \autonodetext{A}[se]{\vspace{4pt}\scriptsize $\mathtt{1}$}  \roundnode{C}(0.5,0.075) \opaquetextfalse  \autonodetext{C}[se]{\vspace{4pt}\scriptsize $\mathtt{v}$} \end{graph},\exists( \begin{graph}(0.8,0.3)(0,0) \fillednodes \roundnode{A}(0.2,0.075) \opaquetextfalse \autonodetext{A}[se]{\vspace{4pt}\scriptsize $\mathtt{1}$}  \roundnode{C}(0.5,0.075) \opaquetextfalse  \autonodetext{C}[se]{\vspace{4pt}\scriptsize $\mathtt{v}$} \end{graph} \mid \gamma_1)) \wedge\ \forall(\begin{graph}(0.7,0.3)(0,0) \fillednodes \roundnode{A}(0.2,0.075) \opaquetextfalse \autonodetext{A}[se]{\vspace{4pt}\scriptsize $\mathtt{1}\mt{=}\mathtt{v}$}  \end{graph},\exists( \begin{graph}(0.7,0.3)(0,0) \fillednodes \roundnode{A}(0.2,0.075) \opaquetextfalse \autonodetext{A}[se]{\vspace{4pt}\scriptsize $\mathtt{1}\mt{=}\mathtt{v}$}  \end{graph} \mid \gamma_1))$ \\

			&& \hspace{0.25in}$\wedge\ \forall( \begin{graph}(1.1,0.3)(0,0) \fillednodes \roundnode{A}(0.2,0.075) \opaquetextfalse \autonodetext{A}[se]{\vspace{4pt}\scriptsize $\mathtt{1}$}  \roundnode{C}(0.5,0.075) \opaquetextfalse  \autonodetext{C}[se]{\vspace{4pt}\scriptsize $\mathtt{v}$} \roundnode{D}(0.8,0.075) \opaquetextfalse  \autonodetext{D}[se]{\vspace{4pt}\scriptsize $\mathtt{w}$} \end{graph} , \exists(\begin{graph}(1.4,0.3)(0,0) \fillednodes \roundnode{A}(0.2,0.075) \opaquetextfalse \autonodetext{A}[se]{\vspace{4pt}\scriptsize $\mathtt{1}$} \roundnode{C}(0.5,0.075) \opaquetextfalse  \autonodetext{C}[se]{\vspace{4pt}\scriptsize $\mathtt{v}$} \roundnode{D}(1.1,0.075) \opaquetextfalse  \autonodetext{D}[se]{\vspace{4pt}\scriptsize $\mathtt{w}$} \diredge{C}{D} \end{graph}) \Rightarrow \exists( \begin{graph}(1.1,0.3)(0,0) \fillednodes \roundnode{A}(0.2,0.075) \opaquetextfalse \autonodetext{A}[se]{\vspace{4pt}\scriptsize $\mathtt{1}$}  \roundnode{C}(0.5,0.075) \opaquetextfalse  \autonodetext{C}[se]{\vspace{4pt}\scriptsize $\mathtt{v}$} \roundnode{D}(0.8,0.075) \opaquetextfalse  \autonodetext{D}[se]{\vspace{4pt}\scriptsize $\mathtt{w}$} \end{graph} \mid \gamma_2))\ $\\
			
			&& \hspace{0.25in}$\wedge\ \forall( \begin{graph}(1,0.3)(0,0) \fillednodes \roundnode{A}(0.2,0.075) \opaquetextfalse \autonodetext{A}[se]{\vspace{4pt}\scriptsize $\mathtt{1}\mt{=}\mathtt{v}$}  \roundnode{C}(0.7,0.075) \opaquetextfalse  \autonodetext{C}[se]{\vspace{4pt}\scriptsize $\mathtt{w}$} \end{graph} , \exists(\begin{graph}(1.1,0.3)(0,0) \fillednodes \roundnode{A}(0.2,0.075) \opaquetextfalse \autonodetext{A}[se]{\vspace{4pt}\scriptsize $\mathtt{1}\mt{=}\mathtt{v}$} \roundnode{B}(0.8,0.075) \opaquetextfalse \autonodetext{B}[se]{\vspace{4pt}\scriptsize $\mathtt{w}$} \diredge{A}{B}   \end{graph}) \Rightarrow \exists( \begin{graph}(1,0.3)(0,0) \fillednodes \roundnode{A}(0.2,0.075) \opaquetextfalse \autonodetext{A}[se]{\vspace{4pt}\scriptsize $\mathtt{1}\mt{=}\mathtt{v}$}  \roundnode{C}(0.7,0.075) \opaquetextfalse  \autonodetext{C}[se]{\vspace{4pt}\scriptsize $\mathtt{w}$} \end{graph} \mid \gamma_2))\ $\\

			&& \hspace{0.25in}$\wedge\ \forall( \begin{graph}(1,0.3)(0,0) \fillednodes \roundnode{A}(0.2,0.075) \opaquetextfalse \autonodetext{A}[se]{\vspace{4pt}\scriptsize $\mathtt{1}\mt{=}\mathtt{w}$}  \roundnode{C}(0.7,0.075) \opaquetextfalse  \autonodetext{C}[se]{\vspace{4pt}\scriptsize $\mathtt{v}$} \end{graph} , \exists(\begin{graph}(1.1,0.3)(0,0) \fillednodes \roundnode{A}(0.2,0.075) \opaquetextfalse \autonodetext{A}[se]{\vspace{4pt}\scriptsize $\mathtt{1}\mt{=}\mathtt{w}$}  \roundnode{C}(0.8,0.075) \opaquetextfalse  \autonodetext{C}[se]{\vspace{4pt}\scriptsize $\mathtt{v}$} \diredge{C}{A} \end{graph}) \Rightarrow \exists( \begin{graph}(1,0.3)(0,0) \fillednodes \roundnode{A}(0.2,0.075) \opaquetextfalse \autonodetext{A}[se]{\vspace{4pt}\scriptsize $\mathtt{1}\mt{=}\mathtt{w}$}  \roundnode{C}(0.7,0.075) \opaquetextfalse  \autonodetext{C}[se]{\vspace{4pt}\scriptsize $\mathtt{v}$} \end{graph} \mid \gamma_2))\ $\\
			
			&& \hspace{0.25in}$\wedge\ (\forall( \begin{graph}(0.7,0.3)(0,0) \fillednodes \roundnode{A}(0.2,0.075) \opaquetextfalse \autonodetext{A}[se]{\vspace{4pt}\scriptsize $\mathtt{1}\mt{=}\mathtt{v}$}  \end{graph} ,  \exists( \begin{graph}(0.7,0.3)(0,0) \fillednodes \roundnode{A}(0.2,0.075) \opaquetextfalse \autonodetext{A}[se]{\vspace{4pt}\scriptsize $\mathtt{1}\mt{=}\mathtt{v}$}  \end{graph} \mid \mathtt{not\ v\in X}))$\\
			
			&& \hspace{0.375in}$\vee\ \forall( \begin{graph}(0.7,0.3)(0,0) \fillednodes \roundnode{A}(0.2,0.075) \opaquetextfalse \autonodetext{A}[se]{\vspace{4pt}\scriptsize $\mathtt{1}\mt{=}\mathtt{v}$}  \end{graph} ,  \exists( \begin{graph}(0.7,0.3)(0,0) \fillednodes \roundnode{A}(0.2,0.075) \opaquetextfalse \autonodetext{A}[se]{\vspace{4pt}\scriptsize $\mathtt{1}\mt{=}\mathtt{v}$}  \end{graph} \mid \mathtt{not\ v\in Y})))\ \mt{]})$ 

		\end{tabular}
	\caption{Trees are 2-colourable}\label{eg:trees_2col}
\end{figure}

	Observe that $\text{Pre}(\mathtt{grow},col)$ is essentially an ``embedding'' of the postcondition $col$ within the context of possible matches for $\tt grow$. The second line expresses that every node (whether the node of the match or not) is coloured $\mathtt{X}$ or $\mathtt{Y}$. The following three conjuncts then express that any edges in the various contexts of the match connect nodes that are differently coloured. The final conjunct is of the same form, but is ``pre-empting'' the creation of a node and edge by $\tt grow$. To ensure that the graph remains 2-colourable, node $\mathtt{1}$ of the match must not belong to both sets; this, of course, is already established by the first nested conjunct. Hence the first implication arising from instances of [cons], $col \Rightarrow \text{Pre}(\mathtt{grow},col)$, is valid. The second implication, $emp \Rightarrow \text{Pre}(\mathtt{init},col)$, is also valid since a graph satisfying $emp$ will not have any nodes to quantify over.
	\qed
\end{example}

	\begin{example}\rm	
		An \emph{acyclic graph} is a graph that does not contain any \emph{cycles}, i.e.\ non-empty paths starting and ending on the same node. One way to test for acyclicity is to apply the rule $\mathtt{delete} = \langle \langle \begin{graph}(1.1,0.3)(0,0) \fillednodes  \roundnode{A}(0.2,0.075) \opaquetextfalse \autonodetext{A}[se]{\vspace{4pt}\scriptsize $\mathtt{1}$} \roundnode{B}(0.8,0.075) \opaquetextfalse \autonodetext{B}[se]{\vspace{4pt}\scriptsize $\mathtt{2}$} \diredge{A}{B} 
		 \end{graph} \Rightarrow 	\begin{graph}(0.8,0.3)(0,0) \fillednodes  \roundnode{A}(0.2,0.075) \opaquetextfalse \autonodetext{A}[se]{\vspace{4pt}\scriptsize $\mathtt{1}$} \roundnode{B}(0.5,0.075) \opaquetextfalse \autonodetext{B}[se]{\vspace{4pt}\scriptsize $\mathtt{2}$}  
			 \end{graph}\rangle, \text{ac}_L \rangle$ for as long as possible; the resulting graph being edgeless if the input graph was acyclic. Here, $\text{ac}_L$ denotes the left application condition $	\neg\exists(\begin{graph}(1.1,0.3)(0,0) \fillednodes  \roundnode{A}(0.2,0.075) \opaquetextfalse \autonodetext{A}[se]{\vspace{4pt}\scriptsize $\mathtt{1}$} \roundnode{B}(0.8,0.075) \opaquetextfalse \autonodetext{B}[se]{\vspace{4pt}\scriptsize $\mathtt{2}$} \diredge{A}{B} 
			 \end{graph} \hookrightarrow \begin{graph}(1.7,0.3)(0,0) \fillednodes  \roundnode{A}(0.8,0.075) \opaquetextfalse \autonodetext{A}[se]{\vspace{4pt}\scriptsize $\mathtt{1}$} \roundnode{B}(1.4,0.075) \opaquetextfalse \autonodetext{B}[se]{\vspace{4pt}\scriptsize $\mathtt{2}$} \diredge{A}{B} \roundnode{C}(0.2,0.075) \diredge{C}{A} 
			 \end{graph}) \vee \neg\exists(\begin{graph}(1.1,0.3)(0,0) \fillednodes  \roundnode{A}(0.2,0.075) \opaquetextfalse \autonodetext{A}[se]{\vspace{4pt}\scriptsize $\mathtt{1}$} \roundnode{B}(0.8,0.075) \opaquetextfalse \autonodetext{B}[se]{\vspace{4pt}\scriptsize $\mathtt{2}$} \diredge{A}{B} 
			 \end{graph}\hookrightarrow \begin{graph}(1.7,0.3)(0,0) \fillednodes  \roundnode{A}(0.2,0.075) \opaquetextfalse \autonodetext{A}[se]{\vspace{4pt}\scriptsize $\mathtt{1}$} \roundnode{B}(0.8,0.075) \opaquetextfalse \autonodetext{B}[se]{\vspace{4pt}\scriptsize $\mathtt{2}$} \diredge{A}{B} \roundnode{C}(1.4,0.075) \diredge{B}{C}
			 \end{graph})$, expressing that in matches, either the source node has indegree $0$ or the target node has outdegree $0$ (we do not consider the special case of looping edges for simplicity). Note that nodes \emph{within} a cycle would not satisfy this: if a source node has an indegree of $0$ for example, there would be no possibility of an outgoing path ever returning to the same node.
		
		We prove two claims about this rule under iteration: first, that it deletes all edges in an acyclic graph; second, that if applied to a graph containing cycles, the resulting graph would not be edgeless. That is, $\vdash \{\neg c \}\ \mathtt{delete}!\ \{e\}$ and $\vdash \{c\}\ \mathtt{delete}!\ \{\neg e\}$, for M-constraints $c$ (for \underline{c}ycles), $e$ (for \underline{e}dgeless), $\gamma_c = \mathtt{path}\mt{(}\mathtt{v,w,not\ e}\mt{)}\ \mathtt{and}\ \mathtt{path}\mt{(}\mathtt{w,v,not\ e}\mt{)}$, and proofs as in Figure \ref{fig:proof_trees_acyclicity}.

	\begin{figure}[htb]
		\centering

	\vspace{-25pt}{\small\begin{tabular}{ p{0.5\textwidth} p{0.5\textwidth} }
				
					\begin{prooftree}

			\AxiomC{$\{\text{Pre}(\mathtt{delete},\neg c)\}\ \mathtt{delete}\ \{\neg c\}$}

			\UnaryInfC{$ \{\neg c\}\ \mathtt{delete}\ \{ \neg c \}$}

				\UnaryInfC{$ \{\neg c\}\ \mathtt{delete}!\ \{ \neg c \wedge \neg\text{App}(\{\mathtt{delete}\})\}$}

				\UnaryInfC{$\vdash \{\neg c\}\ \mathtt{delete}!\ \{ e \}$}
				\end{prooftree}

		    &
				
					\begin{prooftree}

						\AxiomC{$\{\text{Pre}(\mathtt{delete}, c)\}\ \mathtt{delete}\ \{ c\}$}

			\UnaryInfC{$ \{ c\}\ \mathtt{delete}\ \{  c \}$}

				\UnaryInfC{$ \{ c\}\ \mathtt{delete}!\ \{  c \wedge \neg\text{App}(\{\mathtt{delete}\})\}$}

				\UnaryInfC{$\vdash \{ c\}\ \mathtt{delete}!\ \{ \neg e \}$}
				\end{prooftree}
			\end{tabular}}
			
			\begin{tabular}{r c l}
				$c$ &$=$& $\exists (\begin{graph}(0.8,0.3)(0,0) \fillednodes \roundnode{A}(0.2,0.075)  \autonodetext{A}[se]{\vspace{4pt}\scriptsize $\mathtt{v}$} \roundnode{B}(0.5,0.075) \autonodetext{B}[se]{\vspace{4pt}\scriptsize $\mathtt{w}$} \end{graph} \mid \mathtt{path}\mt{(}\mathtt{v,w}\mt{)}\ \mathtt{and}\ \mathtt{path}\mt{(}\mathtt{w,v}\mt{)})$\\

				$e$ &$=$& $\neg\exists (\begin{graph}(1.1,0.3)(0,0) \fillednodes  \roundnode{A}(0.2,0.075) \opaquetextfalse \autonodetext{A}[se]{\vspace{4pt}\scriptsize $\mathtt{v}$} \roundnode{B}(0.8,0.075) \opaquetextfalse \autonodetext{B}[se]{\vspace{4pt}\scriptsize $\mathtt{w}$} \diredge{A}{B} 
				 \end{graph})$\\
				
				$\text{Pre}(\mathtt{delete},\neg c )$ &$=$& $\forall (\begin{graph}(1.1,0.3)(0,0) \fillednodes  \roundnode{A}(0.2,0.075) \opaquetextfalse \autonodetext{A}[se]{\vspace{4pt}\scriptsize $\mathtt{1}$} \roundnode{B}(0.8,0.075) \opaquetextfalse \autonodetext{B}[se]{\vspace{4pt}\scriptsize $\mathtt{2}$} \diredge{A}{B} \bowtext{A}{B}{0.2}{\scriptsize $\mathtt{e}$}
				 \end{graph}, \text{ac}_L \Rightarrow$\\

				&&  \hspace{0.125in}$\wedge\ \neg\exists(\begin{graph}(1.7,0.3)(0,0) \fillednodes  \roundnode{A}(0.2,0.075) \opaquetextfalse \autonodetext{A}[se]{\vspace{4pt}\scriptsize $\mathtt{1}$} \roundnode{B}(0.8,0.075) \opaquetextfalse \autonodetext{B}[se]{\vspace{4pt}\scriptsize $\mathtt{2}$} \diredge{A}{B} \bowtext{A}{B}{0.2}{\scriptsize $\mathtt{e}$} \roundnode{C}(1.1,0.075) \opaquetextfalse \autonodetext{C}[se]{\vspace{4pt}\scriptsize $\mathtt{v}$} \roundnode{D}(1.4,0.075) \opaquetextfalse \autonodetext{D}[se]{\vspace{4pt}\scriptsize $\mathtt{w}$}
				 \end{graph} \mid \gamma_c) \wedge \neg\exists(\begin{graph}(1.4,0.3)(0,0) \fillednodes  \roundnode{A}(0.2,0.075) \opaquetextfalse \autonodetext{A}[se]{\vspace{4pt}\scriptsize $\mathtt{1}\mt{=}\mathtt{v}$} \roundnode{B}(0.8,0.075) \opaquetextfalse \autonodetext{B}[se]{\vspace{4pt}\scriptsize $\mathtt{2}$} \diredge{A}{B} \bowtext{A}{B}{0.2}{\scriptsize $\mathtt{e}$} \roundnode{C}(1.1,0.075) \opaquetextfalse \autonodetext{C}[se]{\vspace{4pt}\scriptsize $\mathtt{w}$} 
				 \end{graph} \mid \gamma_c)$\\

				&& \hspace{0.125in}$\wedge \neg\exists(\begin{graph}(1.6,0.3)(0,0) \fillednodes  \roundnode{A}(0.2,0.075) \opaquetextfalse \autonodetext{A}[se]{\vspace{4pt}\scriptsize $\mathtt{1}$} \roundnode{B}(0.8,0.075) \opaquetextfalse \autonodetext{B}[se]{\vspace{4pt}\scriptsize $\mathtt{2}\mt{=}\mathtt{v}$} \diredge{A}{B} \bowtext{A}{B}{0.2}{\scriptsize $\mathtt{e}$} \roundnode{C}(1.3,0.075) \opaquetextfalse \autonodetext{C}[se]{\vspace{4pt}\scriptsize $\mathtt{w}$} 
				 \end{graph} \mid \gamma_c) \wedge \neg\exists(\begin{graph}(1.4,0.3)(0,0) \fillednodes  \roundnode{A}(0.2,0.075) \opaquetextfalse \autonodetext{A}[se]{\vspace{4pt}\scriptsize $\mathtt{1}\mt{=}\mathtt{w}$} \roundnode{B}(0.8,0.075) \opaquetextfalse \autonodetext{B}[se]{\vspace{4pt}\scriptsize $\mathtt{2}$} \diredge{A}{B} \bowtext{A}{B}{0.2}{\scriptsize $\mathtt{e}$} \roundnode{C}(1.1,0.075) \opaquetextfalse \autonodetext{C}[se]{\vspace{4pt}\scriptsize $\mathtt{v}$} 
				 \end{graph} \mid \gamma_c)$\\

				&& \hspace{0.125in}$\wedge\ \neg\exists(\begin{graph}(1.6,0.3)(0,0) \fillednodes  \roundnode{A}(0.2,0.075) \opaquetextfalse \autonodetext{A}[se]{\vspace{4pt}\scriptsize $\mathtt{1}$} \roundnode{B}(0.8,0.075) \opaquetextfalse \autonodetext{B}[se]{\vspace{4pt}\scriptsize $\mathtt{2}\mt{=}\mathtt{w}$} \diredge{A}{B} \bowtext{A}{B}{0.2}{\scriptsize $\mathtt{e}$} \roundnode{C}(1.3,0.075) \opaquetextfalse \autonodetext{C}[se]{\vspace{4pt}\scriptsize $\mathtt{v}$} 
				 \end{graph} \mid \gamma_c) \wedge \neg\exists(\begin{graph}(1.2,0.3)(0,0) \fillednodes  \roundnode{A}(0.2,0.075) \opaquetextfalse \autonodetext{A}[se]{\vspace{4pt}\scriptsize $\mathtt{1}\mt{=}\mathtt{v}$} \roundnode{B}(0.8,0.075) \opaquetextfalse \autonodetext{B}[se]{\vspace{4pt}\scriptsize $\mathtt{2}\mt{=}\mathtt{w}$} \diredge{A}{B} \bowtext{A}{B}{0.2}{\scriptsize $\mathtt{e}$} 
				 \end{graph} \mid \gamma_c)$\\

				&& \hspace{0.125in}$ \wedge\ \neg\exists(\begin{graph}(1.2,0.3)(0,0) \fillednodes  \roundnode{A}(0.2,0.075) \opaquetextfalse \autonodetext{A}[se]{\vspace{4pt}\scriptsize $\mathtt{1}\mt{=}\mathtt{w}$} \roundnode{B}(0.8,0.075) \opaquetextfalse \autonodetext{B}[se]{\vspace{4pt}\scriptsize $\mathtt{2}\mt{=}\mathtt{v}$} \diredge{A}{B} \bowtext{A}{B}{0.2}{\scriptsize $\mathtt{e}$} 
				 \end{graph} \mid \gamma_c))$\\

				$\text{App}(\{\mathtt{delete}\})$ &$=$& $\exists (\begin{graph}(1.1,0.3)(0,0) \fillednodes  \roundnode{A}(0.2,0.075) \opaquetextfalse \autonodetext{A}[se]{\vspace{4pt}\scriptsize $\mathtt{1}$} \roundnode{B}(0.8,0.075) \opaquetextfalse \autonodetext{B}[se]{\vspace{4pt}\scriptsize $\mathtt{2}$} \diredge{A}{B} 
				 \end{graph}, \text{ac}_L)	$\\
			\end{tabular}
	
		\caption{Acyclity (or lack thereof) is invariant}\label{fig:proof_trees_acyclicity}
	\end{figure}

	First, observe that $\text{Pre}(\mathtt{delete},\neg c)$ is essentially an ``embedding'' of the postcondition $\neg c$ within the context of possible matches for $\tt delete$. The path predicates in $\gamma_c$ now additionally assert (as a result of the L transformation) that paths do not include images of edge $\tt e$: this is crucially important for establishing the postcondition because the rule deletes the edge. For space reasons we did not specify $\text{Pre}(\mathtt{delete},c)$, but this can be constructed from $\text{Pre}(\mathtt{delete},\neg c)$ by replacing each $\wedge$ with $\vee$ and removing each $\neg$ in the nested part.
	
	The instances of [cons] give rise to implications that we must show to be valid. First, $\neg c \Rightarrow \text{Pre}(\mathtt{delete},\neg c)$ is valid: a graph satisfying $\neg c$ does not contain any cycles, hence it also does not contain cycles outside of the context of matches for $\tt delete$. Second, $\neg c \wedge \neg\text{App}(\{\mathtt{delete}\}) \Rightarrow e$ is valid: a graph satisfying the antecedent does not contain any cycles and also no pair of incident nodes for which $\text{ac}_L$ holds. If the graph is not edgeless, then there must be some such pair satisfying $\text{ac}_L$; otherwise the edges are within a cycle. Hence the graph must be edgeless, satisfying $e$.
	
	In the second proof tree, $c \Rightarrow \text{Pre}(\mathtt{delete},c)$ is valid. A graph satisfying $c$ contains a cycle: clearly, no edge (with its source and target) in this cycle satisfies $\text{ac}_L$; hence the graph satisfies the consequent, since images of edge $\tt e$ cannot be part of the cycle in the graph. Finally, $c \wedge \neg \text{App}(\{\mathtt{delete}\}) \Rightarrow \neg e$ is valid: if a graph satisfies the antecedent, then it contains a cycle, the edges of which $\tt delete$ will never be applicable to because of $\text{ac}_L$; hence the graph cannot be edgeless, and satisfies $\neg e$.
	\qed
	\end{example}

	\section{Related Work}\label{sec:related_work}
	
	We point to a few related publications addressing the verification of non-local graph properties through proofs / theorem proving and model checking.
	
	Habel and Radke have considered HR conditions \cite{Habel-Radke10a}, an extension of nested conditions embedding hyperedge replacement grammars via graph variables. The formalism is more expressive than MSO logic on graphs (it is able, for example, to express node-counting MSO properties such as ``the graph has an even number of nodes'' \cite{Radke13a}) but it is not yet clear whether an effective construction for weakest liberal preconditions exists. Percebois et al. \cite{Percebois-Strecker-Tran13a} demonstrate how one can verify global invariants involving paths, directly at the level of rules. Rules are modelled with (a fragment of) first-order logic on graphs in the interactive theorem prover Isabelle. Inaba et al. \cite{Inaba-et-al11a} address the verification of type-annotated Core UnCAL -- a query algebra for graph-structured databases -- against input/output graph schemas in MSO. They first reformulate the query algebra itself in MSO, before applying an algorithm that reduces the verification problem to the validity of MSO over trees.
	
	The GROOVE model checker \cite{Ghamarian-Mol-Rensink-Zambon-Zimakova12a} supports rules with paths in the left-hand side, expressed as a regular expression over edge labels. One can specify such rules to match only when some (un)desirable non-local property holds, and then verify automatically that the rule is never applicable. Augur 2 \cite{Koenig-Kozioura08a} also uses regular expressions, but for expressing forbidden paths that should not occur in any reachable graph.
	
	\section{Conclusion}\label{sec:conclusion}
	
	This paper has contributed the means for systematic proofs of graph programs with respect to non-local specifications. In particular, we defined M-conditions, an extension of nested conditions equivalently expressive to MSO logic on graphs, and defined for this assertion language an effective construction for weakest liberal preconditions of rules. We demonstrated the use of this work in some Hoare-style proofs of programs relative to non-local invariants, i.e.\ the existence of 2-colourings, and the existence of arbitrary-length cycles. Some interesting topics for future work include: extending M-conditions and Pre to support other useful predicates (e.g.\ an \emph{undirected} path predicate), adding support for attribution (e.g.\ along the lines of \cite{Poskitt-Plump12a,Poskitt13a}), implementing the construction of Pre, and generalising the resolution- and tableau-based reasoning systems for nested conditions \cite{Pennemann08a,Lambers-Orejas14a} to M-conditions.\\
	
	\noindent \textbf{Acknowledgements.} The research leading to these results has received funding from the
	European Research Council under the European Union's Seventh Framework
	Programme (FP7/2007-2013) / ERC Grant agreement no. 291389.

	\bibliographystyle{splncs03}
	
	\bibliography{references}
	
	\newpage\appendix
	
	\section*{Appendix: Proofs and Semantics}
	
	\section{Expressive Equivalence to MSO Formulae}\label{app:proofs:expressiveness}
	
	In this section we prove that M-conditions and MSO formulae on graphs are equivalently expressive. We define a many-sorted MSO logic on graphs (in the spirit of \cite{Courcelle90a}), and show that there are translations from this logic to M-conditions and vice versa. The logic and translations are based on those of \cite{Poskitt13a} for nested conditions with expressions. (An alternative approach is to use a single-sorted logic, e.g. \cite{Habel-Radke10a}.) Throughout this section we will assume that graphs are labelled over some fixed label alphabet $\mathcal{C} = \langle \mathcal{C}_V,\mathcal{C}_E \rangle$.
	
	\subsection{Syntax and Semantics}
	
	We define the syntax and semantics of a many-sorted MSO logic on graphs. The idea is to assign sorts (or types) -- edge, vertex, edge set, or vertex set -- to every expression of the logic, and prevent at the syntactic level the composition of formulae that do not ``make sense'' under interpretation. For example, we discard as syntactically ill-formed any expression $\mathtt{s}\mt{(}x\mt{)}$ in which $x$ is not an edge expression (since this will be interpreted as the source function of some graph).

	\begin{definition}[Expressions]\label{defn:ms_logic:expressions}\rm
		The grammar in Figure \ref{fig:app:expressions} defines four syntactic categories of \emph{expressions}: Edge, Vertex, EdgeSet, and VertexSet. They respectively contain (disjoint) syntactic categories of \emph{variables}: EVar, VVar, ESetVar, and VSetVar.

	\vspace*{-.25\baselineskip} 
	\begin{figure}[htb]
	\renewcommand{\arraystretch}{1.2}
	\begin{center}
	\begin{tabular}{lcl}
	Expression & ::= & Edge $\mid$ Vertex $\mid$ EdgeSet $\mid$ VertexSet \\
	Edge & ::= & EVar \\
	Vertex & ::= & VVar $\mid$ ($\mathtt{s}$ $\mid$ $\mathtt{t}$) '$\mt{(}$' Edge '$\mt{)}$'\\
	EdgeSet & ::= & ESetVar \\
	VertexSet & ::= & VSetVar \\
	\end{tabular}
	\end{center}
	\caption{Abstract syntax of expressions \label{fig:app:expressions}}
	\end{figure}
	\qed
	\end{definition}

	\begin{definition}[Sorts, sort function]\label{defn:app:sorts}\rm
		Every expression is associated with a \emph{sort} (or \emph{type}), determined by the syntactic category it is contained within. We use the name of that category to denote its sort. The function $\text{sort}(e)$ is the \emph{sort function}, that takes an expression $e$ as input and returns its sort.
	\qed
	\end{definition}
	
	The formulae of the logic can quantify over first-order and MSO (i.e. set) variables, and express the existence of edges and nodes in sets of the corresponding type. Note that we do not include equality of set variables, since this can be defined precisely in terms of set membership over individual elements.

	\begin{definition}[Formulae]\label{defn:app:formulae}\rm
		Figure \ref{fig:app:ms_formula_syntax} defines \emph{formulae}, where $b \in \mathcal{C}_E$ and $c \in \mathcal{C}_V$.

	\begin{figure}[htb]
	\renewcommand{\arraystretch}{1.2}
	\begin{center}
	\begin{tabular}{lcl}
	Formula & ::= & $\mathtt{true}$ $\mid$ $\mathtt{false}$ $\mid$ Edge '\verb#=#' Edge $\mid$ Vertex '\verb#=#' Vertex\\
	&& $\mid$ $\mathtt{lab}_b$ '$\mt{(}$' Edge '$\mt{)}$' $\mid$ $\mathtt{lab}_c$ '$\mt{(}$' Vertex '$\mt{)}$' \\
	&& $\mid$ Edge '$\in$' EdgeSet $\mid$ Vertex '$\in$' VertexSet \\
	&& $\mid$\ '$\neg$' Formula $\mid$ Formula BoolOp Formula  \\
	&& $\mid$\ Quantifier (VVar '\texttt{:}$\mathtt{V}$' $\mid$ EVar '\texttt{:}$\mathtt{E}$'\\
	&& \hspace{0.25in} $\mid$\ ESetVar '\texttt{:}$\mathtt{ES}$' $\mid$ VSetVar '\texttt{:}$\mathtt{VS}$' ) '\mt{.}' Formula \\
	BoolOp & ::= & $\wedge$ $\mid$ $\vee$ $\mid$ $\Rightarrow$ $\mid$ $\Leftrightarrow$ \\
	Quantifier & ::= & $\forall$ $\mid$ $\exists$
	\end{tabular}
	\end{center}
	\caption{Abstract syntax of formulae}\label{fig:app:ms_formula_syntax}
	\end{figure}

	\qed
	\end{definition}
	
The symbols $\mathtt{s}$, $\mathtt{t}$ are \emph{function symbols} of arity one, and are syntactic representations of source and target functions. The symbols $\text{lab}_y$ are \emph{predicate symbols} of arity one, expressing that an item is labelled by $y$. The symbols $\mt{=}, \in$ are predicate symbols of arity two, and are syntactic representations of equality and set membership.
	
	The \emph{free variables} of a formula are those that are not bound by a quantifier. Note that such variables still have sorts. If a formula contains no such free variables, then we call it a sentence.

	\begin{definition}[Sentence]\label{defn:sentence}\rm
		A \emph{sentence} (or a \emph{closed formula}) is a formula that contains no free variables.
	\qed
	\end{definition}
	
	Sentences of the logic are evaluated with respect to interpretations. These map the sorts to disjoint semantic domains, function symbols to functions, and predicate symbols to Boolean-valued functions. In particular, given some graph, we build an interpretation from its nodes, edges, source, target, and labelling functions. (Note that interpretations here are different from interpretations for M-conditions, which map only set variables to elements of the corresponding semantic domains.)

	\begin{definition}[Satisfaction of sentences]\label{defn:app:graph_structure}\label{defn:app:interpretation_function}\label{defn:app:satisfaction_structures}\rm
		An \emph{interpretation} $I$ is a mapping from (1) sorts to semantic domains, (2) expressions $f(e_1,\dots,e_n)$, with $f$ a function symbol and each $e_i$ an expression, to functions of arity:
	\[I(\text{sort}(e_1)) \times \cdots \times I(\text{sort}(e_n)) \rightarrow I(\text{sort}(f(e_1,\dots,e_n))),\]

	\noindent and (3) formulae $p(e_1,\dots,e_n)$, with $p$ a predicate symbol and each $e_i$ an expression, to Boolean-valued functions of arity:
	\[I(\text{sort}(e_1)) \times \cdots \times I(\text{sort}(e_n)) \rightarrow \mathbb{B}.\]

		Let $I$ be an interpretation function, and $\varphi$ be a sentence. The \emph{satisfaction} of $\varphi$ by $I$, denoted $I \models \varphi$, is defined inductively as follows.

		If $\varphi$ is $\mathtt{true}$ (resp. $\mathtt{false}$), then $I \models \varphi$ (resp. $I \models \varphi$ does not hold). If $\varphi$ is $p(e_1,\dots,e_n)$ with $p$ a predicate symbol and each $e_i$ an expression, then $I \models \varphi$ if $I(p)(I(e_1),\dots,I(e_n)) = \text{true}$.

		Let $\varphi_1,\varphi_2$ be sentences. If $\varphi$ is $\neg\varphi_1$, then $I \models \varphi$ if $I \models \varphi_1$ does not hold. If $\varphi$ is $\varphi_1 \wedge \varphi_2$ (resp. $\varphi_1 \vee \varphi_2$), then $I\models \varphi$ if $I\models \varphi_1$ and (resp. or) $I \models \varphi_2$. If $\varphi$ is $\varphi_1 \Rightarrow \varphi_2$, then $I\models \varphi$ if $I \models \neg \varphi_1$ or $I \models \varphi_2$. If $\varphi$ is $\varphi_1 \Leftrightarrow \varphi_2$, then $I\models \varphi$ if $I\models \varphi_1 \Rightarrow \varphi_2$ and $I\models \varphi_2 \Rightarrow \varphi_1$.

		Let $\mathtt{x}$ be a variable of sort $s$, and $\varphi_1$ be a formula with $\mathtt{x}$ as its only free variable. Let also $S$ denote the symbol that corresponds with sort $s$. If $\varphi$ has the form $\exists \mathtt{x}:S\mt{.}\ \varphi_1$, then $I \models \varphi$ if there is some $a\in I(s)$ such that $I_{\mathtt{x}\mapsto a} \models \varphi_1$ where $I_{\mathtt{x}\mapsto a}$ is equal to $I$ but with the addition that $I(\mathtt{x}) = a$. If $\varphi$ is $\forall \mathtt{x}:S\mt{.}\ \varphi_1$, then $I \models \varphi$ if for every $a\in I(s)$, $I_{\mathtt{x}\mapsto a} \models \varphi_1$.

	\qed
	\end{definition}

	\begin{definition}[Satisfaction of sentences by graphs]\label{defn:interpretation_function}\label{defn:satisfaction_graphs}\rm
		Let $G$ be a graph and $\varphi$ be a sentence. We say that $G$ \emph{satisfies} $\varphi$, denoted by $G \models \varphi$, if $I_G \models \varphi$, where $I_G$ is the \emph{interpretation induced by $G$}, defined as follows:

		\emph{Sorts.} We define $I_G(\text{Edge})=E_G$, $I_G(\text{Vertex})=V_G$, $I_G(\text{EdgeSet}) = 2^{E_G}$, and  $I_G(\text{VertexSet}) = 2^{V_G}$.
		
		\emph{Function symbols.} We define $I_G(\mathtt{s}) = s_G$ and $I_G(\mathtt{t}) = t_G$. We define $I_G(\mathtt{l})$ and $I_G(\mathtt{m})$ to be the functions $l_G$ and $m_G$ respectively.

		\emph{Predicate symbols.} We define $I_G(\mathtt{lab}_b) = \text{lab}_b$ where $\text{lab}_b\!: E_G \rightarrow \mathbb{B}$ returns true for inputs $e$ if $m_G(e) = b$; false otherwise. (Analogous for node label predicates.) We define $I_G(\mt{=})$ to be equality in the standard sense. We define $I_G(\in) = \text{in}_G$ where $\text{in}_G\!:(E_G\times 2^{E_G})\cup (V_G\times 2^{V_G})\rightarrow \mathbb{B}$ returns true for inputs $(x,X)$ if $x \in X$; false otherwise.

	\qed
	\end{definition}

	\subsection{From Formulae to M-Conditions}
	
	In this subsection we prove that formulae can be translated into equivalent M-conditions. We define a translation over the abstract syntax of formulae and expressions. It is assumed that distinct quantifiers bind distinct variables in formulae, allowing us to use node and edge variables as identifiers in the corresponding M-condition. This correspondence is very important in the translation: a node variable $\mathtt{v}$ will correspond to a node identifier $v$ in the M-condition, and an edge variable $\mathtt{e}$ will correspond to an edge identifier $e$ with source and target nodes $s_e, t_e$.
	
	First, we define a helper function that takes a Vertex-sorted expression as input, and returns the node identifier that will be associated with it in the M-condition.

	\begin{definition}[Helper function VertexID]\label{def:VertexID}\rm
		Let $t$ denote an expression in Vertex. We define:

			\[\text{VertexID}(t) = \left\{
				     \begin{array}{ll}
				       v & \text{if $t = \mathtt{v}$ with $\mathtt{v}\in\text{VVar}$}\\
				       s_e & \text{if $t = \mathtt{s}\mt{(}\mathtt{e}\mt{)}$ with $\mathtt{e}\in\text{EVar}$}\\
				       t_e & \text{if $t = \mathtt{t}\mt{(}\mathtt{e}\mt{)}$ with $\mathtt{e}\in\text{EVar}$}\\
				     \end{array}
				   \right.\]
		\qed
	\end{definition}
	
	\begin{theorem}[Sentences can be expressed as M-constraints]\label{thm:Cond}\rm
		Let $\varphi$ denote a sentence. There is a transformation Cond such that for all graphs $G$,

	\[ G \models \varphi\ \ \text{if and only if}\ \ G \models \text{Cond}(\varphi). \]

	\noindent \emph{Construction.} We assume that quantifiers in $\varphi$ bind distinct variables (otherwise one can always rename the variables), which allows for variables to correspond to node and edge identifiers. For all sentences $\varphi$, let $\text{Cond}(\varphi) = \text{Cond}'(\varphi,\emptyset)$. The transformation $\text{Cond}'$ takes the formula that remains to be translated as its first input, and the domain of the next morphism in the generated M-condition as its second input. We define it inductively over the abstract syntax of formulae (Figure \ref{fig:app:ms_formula_syntax}) and expressions (Figure \ref{fig:app:expressions}).

	Let $X$ denote a graph over $\mathcal{C}$. Let $\varphi',\varphi_1,\varphi_2$ denote formulae (not necessarily sentences).

	If $\varphi = \mathtt{true}$ (resp.\ $\mathtt{false}$), then $\text{Cond}'(\varphi,X) = \mathtt{true}$ (resp.\ $\mathtt{false}$). If $\varphi = \neg\varphi'$, then $\text{Cond}'(\varphi,X) = \neg \text{Cond}'(\varphi',X)$. If $\varphi = \varphi_1\oplus\varphi_2$ with $\oplus\in\text{BoolOp}$, then $\text{Cond}'(\varphi,X) = \text{Cond}'(\varphi_1,X)\oplus\text{Cond}'(\varphi_2,X)$.
	
	If $\varphi = \mathtt{e}\ \mt{=}\ \mathtt{f}$ with $\mathtt{e},\mathtt{f}\in\text{EVar}$, then $\text{Cond}'(\varphi,X) = \mathtt{true}$ if edges $e,f$ are identified in $X$, otherwise false.

	If $\varphi = v_1\ \mt{=}\ v_2$ with $v_1,v_2$ in Vertex, then $\text{Cond}'(\varphi,X) = \mathtt{true}$ if $\text{VertexID}(v_1)$, $\text{VertexID}(v_2)$ are identified in $X$, otherwise $\mathtt{false}$.
	
	If $\varphi = \text{lab}_b\mt{(}\mathtt{e}\mt{)}$ with $b\in\mathcal{C}_E$ and $\mathtt{e}$ in EVar, then $\text{Cond}'(\varphi,X) = \mathtt{true}$ if $m_X(e) = b$, otherwise false. (Analogous for node label predicates.)

	If $\varphi = \mathtt{e} \in \mathtt{E}$ with $\mathtt{e}$ in EVar and $\mathtt{E}$ in ESetVar, then:
	\[\text{Cond}'(\varphi,X) = \exists (X\hookrightarrow X\mid e \in \mathtt{E}).\]
		
	If $\varphi = v \in \mathtt{V}$ with $v$ in Vertex and $V$ in VSetVar, then:
		\[\text{Cond}'(\varphi,X) = \exists (X\hookrightarrow X\mid \text{VertexID}(v) \in \mathtt{V}).\]

	If $\varphi = \exists \mathtt{v}:\mathtt{V}\mt{.}\ \varphi'$, then:

	\[ \text{Cond}'(\varphi,X) = \bigvee_{X'\in\text{VMerge}(X,\mathtt{v})} \hspace{-0.25in}\exists(X\hookrightarrow X', \text{Cond}'(\varphi',X')) \]

	\noindent Here, $\text{VMerge}(X,\mathtt{v})$ is the (finite) set of graphs constructed from $X$ by disjointly adding a single node $v$ with some label in $\mathcal{C}_V$, and every graph obtainable from these by identifying a node with $v$.

	If $\varphi = \exists \mathtt{e}:\mathtt{E}\mt{.}\ \varphi'$, then:

	\[ \text{Cond}'(\varphi,X) =  \bigvee_{X'\in\text{EMerge}(X,\mathtt{e})} \hspace{-0.25in}\exists ( X\hookrightarrow X',\text{Cond}'(\varphi',X')) \]

	\noindent Here, $\text{EMerge}(X,\mathtt{e})$ is the (finite) set of graphs defined as follows. Let $X^*$ denote a graph obtained from $X$ by disjointly adding nodes with identifiers $s_{e}$, $t_{e}$ and an edge with identifier $e$ such that $s_{X^*}(e) = s_{e}$, $t_{X^*}(e) = t_{e}$, $l_{X^*}(s_{e}) \in \mathcal{C}_V$, $l_{X^*}(t_{e}) \in \mathcal{C}_V$, and $m_{X^*}(e) \in \mathcal{C}_E$. The set $\text{EMerge}(X,\mathtt{e})$ contains all such graphs $X^*$, and all other graphs obtainable from them by identifying $e,s_{e},t_{e}$ with nodes and edges. (Note that $s_e$ and $t_e$ can be identified to create a loop.)
	
	If $\varphi = \exists \mathtt{E}:\mathtt{ES}\mt{.}\ \varphi'$, then:
	
	\[ \text{Cond}'(\varphi,X) = \exists_\mathtt{E}\mathtt{E} \mt{[}\ \text{Cond}'(\varphi',X)\ \mt{]} \]
	
	If $\varphi = \exists \mathtt{V}:\mathtt{VS}\mt{.}\ \varphi'$, then:
	
	\[ \text{Cond}'(\varphi,X) = \exists_\mathtt{V}\mathtt{V} \mt{[}\ \text{Cond}'(\varphi',X)\ \mt{]} \]

	If $\varphi = \forall \mathtt{x}:S\mt{.}\ \varphi'$, then:
	\[\text{Cond}'(\varphi,X) = \text{Cond}'(\neg\exists \mathtt{x}:S\mt{.}\ \neg\varphi', X).\]
	\qed
	\end{theorem}
	
	To prove the theorem, we first prove a more general lemma about the translation of formulae.

	\begin{lemma}[Formulae can be expressed as M-conditions]\label{lemma:Cond}\rm
		Let $\varphi$ denote a formula, $I$ an interpretation defined for the free set variables of $\varphi$, and $X$ a graph in which every identifier $x$ corresponds to a free (node or edge) variable $\tt x$ in $\varphi$. For all injective graph morphisms $z\!:X\hookrightarrow G$, we have:
		\[ I_G^{z,I} \models \varphi\ \ \text{if and only if}\ \ z\!: X\hookrightarrow G \models^I \text{Cond}'(\varphi,X).\]

		\noindent Here, $I_G^{z,I}$ is defined as $I_G$ but with the following mappings for free variables in $\varphi$: (1) for each set variable $\mathtt{Y}$ in the domain of $I$, $I_G^{z,I} (\mathtt{Y}) = I(\mathtt{Y})$; (2) for each node $v$ in $X$, $I_G^{z,I} (\mathtt{v}) = z(v)$; and (3) for each edge $e$ in $X$, $I_G^{z,I} (\mathtt{e}) = z(e)$.
	\qed
	\end{lemma}
	
	\begin{proof}
		\emph{Basis.} Most cases are easily adapted from the proof of Lemma 6.18 in \cite{Poskitt13a}. In the case that $\varphi$ has the form $\mathtt{lab}_b\mt{(}\mathtt{e}\mt{)}$ with $\mathtt{e}$ in EVar,\\
		
		\begin{center}\begin{tabular}{r c l }
			$(I_G^{z,I} \models \mathtt{lab}_b\mt{(}\mathtt{e}\mt{)} )$ & $=$ & $I_G^{z,I}(\mathtt{lab}_b)(\mathtt{e})$ \\
			
			&$=$& $(\text{lab}_b(z(e)))$  \\
			
			& $=$ & $(m_G(z(e)) = b)$  \\
			
			& $=$ & $(m_X(e) = b)$ \\
			
			&$=$& $(z \models^I \text{Cond}'(\varphi,X))$
		\end{tabular}\end{center}
		
		\noindent (Analogous for node label predicates.)
		
		In the case that $\varphi$ has the form $\mathtt{e} \in \mathtt{E}$ with $\mathtt{e}$ in EVar and $\mathtt{E}$ in ESetVar,
		
		\begin{center}\begin{tabular}{r c l }
			$(I_G^{z,I} \models \mathtt{e} \in \mathtt{E} )$ & $=$ & $I_G^{z,I}(\in)(\mathtt{e},\mathtt{E})$ \\
			
			&$=$& $(\text{in}_G(z(e),I(\mathtt{E})))$  \\
			
			& $=$ & $(z(e) \in I(\mathtt{E}))$  \\
			
			& $=$ & $(e \in \mathtt{E})^{I,z}$ \\
			
			& $=$ & $(z \models^I \exists(X\hookrightarrow X \mid e \in \mathtt{E}))$ \\
			
			&$=$& $(z \models^I \text{Cond}'(\varphi,X))$
		\end{tabular}\end{center}
		
		\noindent \emph{Step. Only if.} Assume that $I_G^{z,I} \models \varphi$. Most cases are easily adapted from the proof of Lemma 6.18 in \cite{Poskitt13a}. In the case that $\varphi$ has the form $\exists\mathtt{E}:\mathtt{ES}\mt{.}\ \varphi'$, by assumption, there exists some $a\in 2^{E_G}$ such that $I_G^{z,I} \cup \{\mathtt{E}\mapsto a\} \models \varphi'$. Define $I' = I \cup \{\mathtt{E}\mapsto a\}$. Then $I_G^{z,I'} \models \varphi'$, and by induction hypothesis, $z \models^{I'} \text{Cond}'(\varphi',X)$. Finally, with the definition of $\models$ for M-conditions, we get the result that $z \models^I \exists_\mathtt{E}\mathtt{E}\mt{[}\ \text{Cond}'(\varphi',X)\ \mt{]} = \text{Cond}'(\varphi,X)$. (Case for node set quantification is analogous.)  \\
		
		\noindent \emph{Step. If.} Assume that $z \models^I \text{Cond}'(\varphi,X)$. Most cases are easily adapted from the proof of Lemma 6.18 in \cite{Poskitt13a}. In the case that $\varphi$ has the form $\exists\mathtt{E}:\mathtt{ES}\mt{.}\ \varphi'$, by assumption and construction, we have $z \models^I \exists_\mathtt{E}\mathtt{E}\mt{[}\ \text{Cond}'(\varphi',X)\ \mt{]}$. Then there is some $I' = I \cup \{\mathtt{E}\mapsto a\}$ with $a\subseteq E_G$ (equiv.\ $a\in 2^{E_G}$) such that $z \models^{I'} \text{Cond}'(\varphi',X)$. By induction hypothesis, we have $I_G^{z,I'} \models \varphi'$. Finally, with the definition of $\models$ for formulae, we have the result that $I_G^{z,I} \models \exists\mathtt{E}:\mathtt{ES}\mt{.}\ \varphi'$. (Case for node set quantification is analogous.)		
		\qed
	\end{proof}
	
	\begin{proof}[of Theorem \ref{thm:Cond}]
		Define $i_G\!:\emptyset \hookrightarrow G$. We have that:
		
		\begin{center}\begin{tabular}{r c l}
			$G \models \varphi$ &$\ \text{iff}\ $& $I_G \models \varphi$ \\
			&$\ \text{iff}\ $& $I_G^{i_G,I_\emptyset} \models \varphi$ \\
			&$\ \text{iff}\ $& $i_G \models^{I_\emptyset} \text{Cond}'(\varphi,\emptyset)$ \\
			&$\ \text{iff}\ $& $i_G \models^{I_\emptyset} \text{Cond}(\varphi)$ \\
			&$\ \text{iff}\ $& $G\models \text{Cond}(\varphi)$.
		\end{tabular}\end{center}
		
		\noindent from the definition of $\models$ for formulae and M-conditions, and Lemma \ref{lemma:Cond}.
		\qed
	\end{proof}

	\subsection{From M-Conditions to Formulae}
	
	In this subsection we prove that M-conditions can be translated into equivalent formulae. To simplify the translation, we define a normal form for M-conditions that allows us to assume one new node or edge per level of nesting, as well as the absence of path predicates.
	
	\begin{definition}[Normal form for M-conditions]\rm
		An M-condition $c$ is in \emph{normal form} if all of the following hold:
		
		\begin{enumerate}
			\item all morphisms are inclusions;
			
			\item all morphisms $a\!:P\hookrightarrow C$ are either identity morphisms ($P=C$), or $C$ is the graph $P$ but with one additional node or one additional edge;
					
			\item no interpretation constraint contains a path predicate.
		\end{enumerate}
		\qed
	\end{definition}
	
	\begin{proposition}[M-conditions can be normalised]\rm
		For every M-condition $c$, there is an M-condition $\overline{c}$ in normal form such that $c$ and $\overline{c}$ are equivalent, i.e.\ for every morphism $p$, and every interpretation $I$,
		\[ p \models^I c\ \ \text{if and only if}\ \ p \models^I \overline{c}. \]
		\qed
	\end{proposition}
	
	\begin{proof}[sketch]\rm
		Section 6.4.1 of \cite{Poskitt13a} shows how morphisms can be replaced and decomposed to satisfy (1) and (2) of normal form. For (3), observe that an M-condition $\exists(a\!:X\hookrightarrow X \mid \mathtt{path}\mt{(}\mathtt{v,w,not\ }e_1\mt{|}e_2\mt{|}\dots                                                                                                                                                        \mt{)})$ can be replaced by the equivalent M-condition $\text{Cond}'(\varphi,X)$, where $\varphi$ is defined as follows:\\
		
		\begin{center}\begin{tabular}{r c l}
			$\varphi$ &$=$& \vspace{5pt}$\forall\mathtt{X}:\mathtt{VS}\mt{.}\ (( (\forall\mathtt{y,z}:\mathtt{V}\mt{.}\ \mathtt{y}\in\mathtt{X} \wedge (\exists\mathtt{e}:\mathtt{E}\mt{.}\ \mathtt{s}\mt{(}\mathtt{e}\mt{)}\ \mt{=}\ \mathtt{y} \wedge \mathtt{t}\mt{(}\mathtt{e}\mt{)}\ \mt{=}\ \mathtt{z}$\\
			
			&& \vspace{5pt}\hspace{0.5in}$\wedge\ \neg\mathtt{e}\mt{=}\mathtt{e}_1 \wedge \neg\mathtt{e}\mt{=}\mathtt{e}_2\wedge\dots)\Rightarrow \mathtt{z}\in\mathtt{X})$\\
			
			&& \vspace{5pt}\hspace{0.25in}$\wedge\ (\forall\mathtt{y}:\mathtt{V}\mt{.}\ (\exists\mathtt{e}:\mathtt{E}\mt{.}\ \mathtt{s}\mt{(}\mathtt{e}\mt{)}\ \mt{=}\ \mathtt{v} \wedge \mathtt{t}\mt{(}\mathtt{e}\mt{)}\ \mt{=}\ \mathtt{y}$\\
			
			&& \vspace{5pt}\hspace{0.5in}$\wedge\ \neg\mathtt{e}\mt{=}\mathtt{e}_1 \wedge \neg\mathtt{e}\mt{=}\mathtt{e}_2\wedge\dots) \Rightarrow \mathtt{y}\in\mathtt{X}))$\\
			
			&& \vspace{5pt}\hspace{0.25in}$\Rightarrow\ \mathtt{w}\in\mathtt{X} )$
		\end{tabular}\end{center}
		
		\noindent i.e.\ path predicates can be expressed in terms of MSO expressions.
			\qed
	\end{proof}
	
	Now, we define and prove the correctness of a translation from M-conditions (in normal form) to formulae. The assumption that morphisms are inclusions allows us to establish a correspondence between identifiers and variables. For example, a node with identifier $v$ will be translated into a variable $\tt v$ from VVar.
	
	\begin{theorem}[M-conditions can be expressed as formulae]\label{thm:Form}\rm
		There is a transformation Form such that for all M-constraints $c$, and all graphs $G$, we have:

		\[ G \models c\ \ \text{if and only if}\ \ G \models \text{Form}(c).\]

		\noindent \emph{Construction.} We assume that M-constraint $c$ is in normal form (otherwise replace it with an equivalent M-constraint that is). Define $\text{Form}(c) = \text{Form}'(c,\{\})$. Here, the second parameter can understood as the set of all node and edge variables that have already been bound to quantifiers by the transformation. Then $\text{Form}'(c,V)$ is defined inductively as follows, where $V$ denotes a set of sorted variables.

		If $c=\mathtt{true}$, then $\text{Form}'(c,V) = \mathtt{true}$. If $c = \exists_\mathtt{V}\mathtt{X}\mt{[}c'\mt{]}$, then $\text{Form}'(c,V) = \exists\mathtt{X}:\mathtt{VS}\mt{.}\ \text{Form}'(c',V)$ (analogous for edge set quantification). If $c=\exists(a\mid\gamma,c')$, then there are three possible outputs for $\text{Form}'(c,V)$ defined for the three forms that $a$ may take in normal form.

		Suppose that $c = \exists(\text{id}_P\!:P\hookrightarrow P\mid \gamma, c')$, i.e.\ an M-condition with a morphism that is an identity. Then, $\text{Form}'(c,V)$ is equal to:\\
		\[ \gamma^* \wedge \text{Form}'(c',V). \]

	\noindent Here (and in the following), $\gamma^*$ denotes the formula obtained from $\gamma$ by replacing node identifiers $v$ (resp.\ edge identifiers $e$) with variables $\tt v$ in VVar (resp.\ $\tt e$ in EVar), and by replacing $\mathtt{and},\mathtt{or},\mathtt{not}$ respectively with $\wedge,\vee,\neg$.

		Suppose that $c = \exists([va]\!:P\hookrightarrow P'\mid \gamma, c')$, where $[va]$ denotes a morphism with codomain $P'$ equal to domain $P$, except for an additional node $v$ labelled with $a\in\mathcal{C}_V$. Then, $\text{Form}'(c,V)$ is equal to:

		\begin{center}
		\begin{tabular}{l}
			\vspace{5pt}$\exists\mathtt{v}:\mathtt{V}\mt{.}\ (\bigwedge_{\mathtt{v}'\in V\cap \text{VVar}} \neg\mathtt{v}\ \mt{=}\ \mathtt{v}') \wedge \mathtt{lab}_a\mt{(}\mathtt{v}\mt{)} \wedge \gamma^* \wedge \text{Form}'(c',V\cup \{\mathtt{v}\})$
		\end{tabular}
		\end{center}

		Suppose that $c = \exists([euva]\!:P\hookrightarrow P'\mid \gamma, c')$, where $[euva]$ denotes a morphism with codomain $P'$ equal to domain $P$, except for an additional edge $e$ with label $a\in\mathcal{C}_E$, source node $u$, and target node $v$. Then, $\text{Form}'(c,V)$ is equal to:

		\begin{center}
		\begin{tabular}{l}
			\vspace{5pt}$\exists\mathtt{e}:\mathtt{E}\mt{.}\ (\bigwedge_{\mathtt{e}'\in V\cap \text{EVar}} \neg\mathtt{e}\ \mt{=}\ \mathtt{e}') \wedge \mathtt{lab}_a\mt{(}\mathtt{e}\mt{)} \wedge \mathtt{s}\mt{(}\mathtt{e}\mt{)}\ \mt{=}\ \mathtt{u} \wedge \mathtt{t}\mt{(}\mathtt{e}\mt{)}\ \mt{=}\ \mathtt{v}$\\
			
			\vspace{5pt}\hspace{0.5in}$\wedge\ \gamma^* \wedge \text{Form}'(c',V\cup \{\mathtt{e}\})$
		\end{tabular}
		\end{center}

	\noindent Note that this exploits the correspondence between node identifiers and node variables established in the previous case.

	For Boolean formulae over M-conditions, the transformation $\text{Form}'$ is defined in the standard way, that is, $\text{Form}'(\neg c,V) = \neg \text{Form}'(c,V)$, $\text{Form}'(c \wedge d,V) = \text{Form}'(c,V) \wedge \text{Form}'(d,V)$, and $\text{Form}'(c \vee d,V) = \text{Form}'(c,V) \vee \text{Form}'(d,V)$.

	\qed
	\end{theorem}
	
	To prove the theorem, we first prove a more general lemma about the translation of M-conditions.
	
	\begin{lemma}[M-conditions can be expressed as formulae]\label{lemma:Form}\rm
		For every M-condition $c$ in normal form, injective graph morphism $p\!:P\hookrightarrow G$, and interpretation $I$ in $G$ (defined for the free variables of $c$), we have:
		\[ p\!:P\hookrightarrow G \models^I c\ \ \text{if and only if}\ \ I_G^{p,I} \models \text{Form}'(c,P^*) \]
		
		\noindent where $P^*$ is the set of node and edge variables corresponding to the identifiers in $P$. Furthermore, $I_G^{p,I}$ is defined as $I_G$ but with the following mappings for free variables in $\varphi$: (1) for each set variable $\mathtt{Y}$ in the domain of $I$, $I_G^{p,I} (\mathtt{Y}) = I(\mathtt{Y})$; (2) for each node $v$ in $X$, $I_G^{p,I} (\mathtt{v}) = p(v)$; and (3) for each edge $e$ in $X$, $I_G^{p,I} (\mathtt{e}) = p(e)$.
		\qed
	\end{lemma}
	
	\begin{proof}
		Cases $c = \exists_\mathtt{V}\mathtt{X}\mt{[}c'\mt{]}$ and $c = \exists_\mathtt{E}\mathtt{X}\mt{[}c'\mt{]}$ are clear from the definition of $\models$, the construction, and induction hypothesis. All other cases are easily adapted from the proof of Lemma 6.33 in \cite{Poskitt13a}.
		\qed
	\end{proof}

	\begin{proof}[of Theorem \ref{thm:Form}]
		Define $i_G\!:\emptyset \hookrightarrow G$. We have that:
		
		\begin{center}\begin{tabular}{r c l}
			$G \models c$ &$\ \text{iff}\ $& $i_G \models^{I_\emptyset} c$ \\
			&$\ \text{iff}\ $& $I_G^{i_G,I_\emptyset}\models \text{Form}'(c,\{\})$ \\
			&$\ \text{iff}\ $& $I_G\models \text{Form}'(c,\{\})$ \\
			&$\ \text{iff}\ $& $G\models \text{Form}(c)$.
		\end{tabular}\end{center}
		
		\noindent from the definition of $\models$ for M-conditions and formulae, and Lemma \ref{lemma:Form}.
		\qed
	\end{proof}
	
	\subsection{Proof of Theorem \ref{thm:formulae_equiv}}\label{thm:formulae_equiv:PROOF}
	
	\begin{proof}
		We obtain the result directly from Theorems \ref{thm:Cond} and \ref{thm:Form}.
		\qed
	\end{proof}

	\newpage\section{Semantics of Graph Programs}\label{app:semantics}
	
	This appendix contains an operational semantics -- in the style of GP 2 \cite{Plump12a} -- for the graph programs defined in this paper. The semantics consists of inference rules, which inductively define a small-step transition relation $\rightarrow$ on \emph{configurations}. Intuitively, configurations represent the current state (a graph or special failure state) paired with a program that remains to be executed.

	\begin{definition}[Configuration]\rm
		Let $\mathcal{P}$ denote the class of all graph programs and $\G$ the set of all graphs over $\mathcal{C}$. A \emph{program configuration} is either a program with a graph in $\mathcal{P}\times \G$, just a graph in $\G$, or the special element fail.	
		\qed
	\end{definition}

	\begin{definition}[Transition relation]\rm
	    A \emph{small-step transition relation}
	\[ \rightarrow\ \subseteq (\mathcal{P}\times\G) \times ((\mathcal{P}\times\G)\cup\G\cup\{\text{fail}\}) \]
	\noindent over configurations defines the individual steps of computation. The transitive and reflexive-transitive closures of $\rightarrow$ are written $\rightarrow^+$ and $\rightarrow^*$ respectively.
		\qed
	\end{definition}

	Configurations in $\mathcal{P}\times \G$ represent states of unfinished computations, whereas graphs in $\G$ are proper results. The configuration fail represents a failure state. A configuration $\gamma$ is said to be \emph{terminal} if there is no configuration $\delta$ such that $\gamma \rightarrow \delta$.

	We provide semantic inference rules for the commands of programs. Each inference rule has a premise and conclusion, separated by a horizontal bar. Both contain (implicitly) universally quantified meta-variables for programs and graphs, where $\R$ stands for a rule set call, $C,P,P',Q$ for programs in $\mathcal{P}$, and $G,H$ for graphs in $\G$.

	\begin{definition}[Semantic inference rules for core commands]\label{def:core_sos_rules}\rm
		The \emph{inference rules for core commands} of programs are given in Figure \ref{fig:core_sos_rules}. The notation $G\not\Rightarrow_\R$ expresses that for a graph $G$, there is no graph $H$ such that $G\Rightarrow_\R H$.

	\begin{figure}[htb]
	\centering
	\begin{tabular}{lcl}
	$\mathrm{[call_1]}_\text{OS}$ $\frac{\displaystyle G \dder_\R H}{\displaystyle\tuple{\R,\,G} \to H}$ 
	&&
	$\mathrm{[call_2]}_\text{OS}$ $\frac{\displaystyle G \not\dder_\R}{\displaystyle\tuple{\R,\,G} \to \failrm}$
	\\\\
	$\mathrm{[seq_1]}_\text{OS}$ $\frac{\displaystyle \tuple{P,\, G} \to \tuple{P',\, H}}{\displaystyle \tuple{P;Q,\, G} \to \tuple{P';Q,\, H}}$ 
	&&
	$\mathrm{[seq_2]}_\text{OS}$ $\frac{\displaystyle \tuple{P,\, G} \to H}{\displaystyle \tuple{P;Q,\, G}\to \tuple{Q,\, H}}$
	\\\\
	$\mathrm{[seq_3]}_\text{OS}$ $\frac{\displaystyle \tuple{P,\, G} \to \failrm}{\displaystyle \tuple{P;Q,\, G}\to \failrm}$
	\\\\
	\multicolumn{3}{c}{$\mathrm{[if_1]}_\text{OS}$ $\frac{\displaystyle \tuple{C,\, G} \to^+ H}{\displaystyle \tuple{\ifte{C}{P}{Q},\, G}\to \tuple{P,\, G}}$}
	\\\\
	\multicolumn{3}{c}{$\mathrm{[if_2]}_\text{OS}$ $\frac{\displaystyle \tuple{C,\, G} \to^+ \failrm}{\displaystyle \tuple{\ifte{C}{P}{Q},\, G} \to \tuple{Q,\, G}}$}
	\\\\
	\multicolumn{3}{c}{$\mathrm{[try_1]}_\text{OS}$ $\frac{\displaystyle \tuple{C,\, G} \to^+ H}{\displaystyle \tuple{\tryte{C}{P}{Q},\, G}\to \tuple{P,\, H}}$}
	\\\\
	\multicolumn{3}{c}{$\mathrm{[try_2]}_\text{OS}$ $\frac{\displaystyle \tuple{C,\, G} \to^+ \failrm}{\displaystyle \tuple{\tryte{C}{P}{Q},\, G} \to \tuple{Q,\, G}}$}
	\\\\
	$\mathrm{[alap_1]}_\text{OS}$ $\frac{\displaystyle \tuple{P,\, G} \to^+ H}{\displaystyle \tuple{P!,\, G} \to \tuple{P!,\, H}}$
	&&
	$\mathrm{[alap_2]}_\text{OS}$ $\frac{\displaystyle \tuple{P,\, G} \to^+ \failrm}{\displaystyle \tuple{P!,\, G} \to G}$
	\end{tabular} 
	\caption{Inference rules for core commands}\label{fig:core_sos_rules}
	\end{figure}
		\qed
	\end{definition}

	To convey an intuition as to how the rules should be read, consider the rule [call$_1$]$_\text{OS}$. This reads: ``for all sets of rules $\R$ and all graphs $G,H$, $G\Rightarrow_\R H$ implies that $\langle \R,G \rangle \rightarrow H$''.

	By inspection of the inference rules, we note that a program execution can only result in a failure state if a set of rules is applied to a graph for which no rule in the set is applicable.

	The meaning of programs is given by the semantic function $\llbracket \_ \rrbracket$, which assigns to each program $P$ the function $\llbracket P \rrbracket$ mapping an input graph $G$ to the set of all possible results of executing $P$ on $G$. The application of function $\llbracket P \rrbracket$ to graph $G$ is denoted $\llbracket P \rrbracket G$. As well as graphs, this set may contain the special values fail and $\bot$. The former indicates a program run ending in failure, whereas $\bot$ indicates that at least one execution diverges (does not terminate), or ``gets stuck''.

	\begin{definition}[Divergence]\rm
		A program $P$ \emph{can diverge from} graph $G$ if there is an infinite sequence:
		\[ \langle P, G \rangle \rightarrow \langle P_1, G_1 \rangle \rightarrow \langle P_2, G_2 \rangle \rightarrow \dots \]
		\qed
	\end{definition}

	\begin{definition}[Getting stuck]\rm
		A program $P$ \emph{can get stuck from} graph $G$ if there is a terminal configuration $\langle Q, H \rangle$ such that $\langle P, G \rangle \rightarrow^* \langle Q, H \rangle$.
		\qed
	\end{definition}

	A program can get stuck if the guard program $C$ of a conditional can diverge on some graph $G$, neither producing a graph nor failing, or if the same property is true for a program that is iterated. The execution in these cases gets stuck because none of the inference rules for conditionals and iteration can be applied.

	\begin{definition}[Semantic function]\label{def:semantic_function}\rm
	The \emph{semantic function} $\llbracket \_ \rrbracket\!: \mathcal{P}\rightarrow (\G \rightarrow 2^{\G \cup \{\text{fail},\bot\}} )$, given a graph $G$ and a program $P$, is defined by:

	\begin{align*}
		\llbracket P \rrbracket G =& \ \{ X \in \G \cup \{\text{fail}\} \mid \langle P, G \rangle \rightarrow^+ X \}\\
		& \ \cup \{ \bot \mid \text{$P$ can diverge or get stuck from $G$}\}.\\
	\end{align*}
	\noindent	\qed
	\end{definition}

	Finally, we provide a straightforward definition of program equivalence which is based on the definition of semantic functions.

	\begin{definition}[Semantic equivalence]\rm
		Two graph programs $P$ and $Q$ are \emph{semantically equivalent}, denoted by $P \equiv Q$, if $\llbracket P \rrbracket = \llbracket Q \rrbracket$.
		\qed
	\end{definition}

	\newpage\section{Weakest Liberal Precondition Constructions}
	
	\subsection{Proof of Lemma \ref{lemma:shifting}}\label{lemma:shifting:PROOF}
	
	\begin{proof}
		By structural induction.\\

		\noindent \emph{Basis.} Let $c= \text{true}$. Then $\text{A}'(p,\text{true}) = \text{true}$. All morphisms satisfy true.\\

		\noindent \emph{Step. Only if.} Let $c = \exists_\mathtt{V}\mathtt{X}\mt{[}c'\mt{]}$. Assume that:
		\[p'' \models^I \text{A}'(p,\exists_\mathtt{V}\mathtt{X}\mt{[}c'\mt{]}) = \exists_\mathtt{V}\mathtt{X}\mt{[}\text{A}'(p,c')\mt{]}.\]
		
		\noindent Then there exists an interpretation $I' = I \cup \{\mathtt{X}\mapsto V\}$ for some $V \subseteq V_H$ such that $p'' \models^{I'} \text{A}'(p,c')$. By induction hypothesis, we have $p'' \circ p \models^{I'} c'$. By definition of $\models$, we get the result that $p'' \circ p \models^I \exists_\mathtt{V}\mathtt{X}\mt{[}c'\mt{]}$. Analogous for case $c = \exists_\mathtt{E}\mathtt{X}\mt{[}c'\mt{]}$.

		Let $c = \exists(a \mid \gamma,c')$. Assume that:
		\[p'' \models^{I} \text{A}'(p,\exists(a \mid \gamma,c')) = \bigvee_{e\in\varepsilon}\exists(b\mid\gamma,\text{A}'(s,c'))\]
		
		\noindent i.e.\ there exists an $e\in\varepsilon$ such that $p'' \models^I \exists(b\mid\gamma,\text{A}'(s,c'))$. By the definition of $\models$, there exists a morphism $q''\!: E \hookrightarrow H$ with $p'' = q'' \circ b$, $\gamma^{I,q''} = \text{true}$, and $q''\models^I \text{A}'(s,c')$. Directly from the proof of Lemma 3 in \cite{Habel-Pennemann09a}, we get $p''\circ p \models^{I} \exists (a)$. Using the induction hypothesis, $q''\models^I \text{A}'(s,c')$ implies $q'' \circ s \models^{I} c'$. Define $q'\!: C\hookrightarrow H$ as $q'' \circ s$. With the definition of $\models$, we have $p''\circ p \models^{I} \exists(a,c')$. Finally, with the assumption that $\gamma^{I,q''} = \text{true}$, that $\gamma$ is defined only for nodes and edges in $C$, and that $q' = q'' \circ s$, we have that $\gamma^{I,q'} = \text{true}$. Together, we have the result that $p'' \circ p \models^{I} \exists(a\mid \gamma,c') = c$.\\

		\noindent \emph{Step. If.} Let $c = \exists_\mathtt{V}\mathtt{X}\mt{[}c'\mt{]}$. Assume that $p''\circ p \models^I \exists_\mathtt{V}\mathtt{X}\mt{[}c'\mt{]}$. Then there exists an interpretation $I' = I \cup \{X\mapsto V\}$ for some $V\subseteq V_H$ such that $p''\circ p \models^{I'} c'$. By induction hypothesis, we have that $p'' \models^{I'} \text{A}'(p,c')$. By definition of $\models$ and the construction of $\text{A}'$, we get the result that $p'' \models^I \exists_\mathtt{V}\mathtt{X}\mt{[}\text{A}'(p,c')\mt{]} = \text{A}'(p,\exists_\mathtt{V}\mathtt{X}\mt{[}c'\mt{]})$. Analogous for case $c = \exists_\mathtt{E}\mathtt{X}\mt{[}c'\mt{]}$.

		Let $c = \exists(a \mid \gamma, c')$. As for the ``only if'' direction, one can derive $p'' \models^I \text{A}'(p,\exists(a\mid \gamma,c'))$ directly from the proof of Lemma 3 in \cite{Habel-Pennemann09a}, with the additional requirement that $\gamma^{I,q'} = \gamma^{I,q''} = \text{true}$ for the satisfying morphisms $q'\!:C\hookrightarrow H$ and $q''\!:E\hookrightarrow H$ (simple to show, because $\gamma$ is defined only over items in $C$, and the proof in \cite{Habel-Pennemann09a} shows that $q' = q'' \circ s$).\\

		For Boolean formulae over M-conditions, the statement follows from the definition of $\models$ and the induction hypothesis.
		\qed
	\end{proof}

	\newpage\subsection{Proof of Proposition \ref{prop:LPath}}\label{prop:LPath:PROOF}
	
	\begin{proof}
		Let $p = \mathtt{path}\mt{(}v,w,\mathtt{not}\ E\mt{)}$ with $v,w$ denoting some nodes in $V_R$, and $E$ a set of edges in $E_R$. Furthermore, given some morphism $p$, let $p(E)$ abbreviate the set $\{p(e) \mid e\in E\}$. We show that the equality holds for all contexts of $v,w$ and all types of rules $r$ (in the sense of what the rules create and/or delete).\\
		
		\noindent\emph{Case (1).} Suppose that $v,w\in V_K$ and $E_L = E_R$. Then $\text{LPath}(r,p)$ returns the predicate $\mathtt{path}\mt{(}v,w,\mathtt{not}\ E\mt{)}$. By the definition of interpretations, the equality holds if $\left[\text{path}_G(g(v),g(w),g(E))\ \text{iff}\ \text{path}_H(h(v),h(w),h(E))\right]$. If such a path exists in $G$ then the same path exists in $H$ (and vice versa), since the rule does not create or delete edges, and since any nodes created or deleted would not be part of such paths (by the dangling condition).\\
		
		\noindent\emph{Case (2).} Suppose that $v,w\in V_K$ and $E_R \subset E_L$. Then $\text{LPath}(r,p)$ returns the predicate $\mathtt{path}\mt{(}v,w,\mathtt{not}\ E^\ominus\mt{)}$. By the definition of interpretations, the equality holds if $\left[\text{path}_G(g(v),g(w),g(E^\ominus))\ \text{iff}\ \text{path}_H(h(v),h(w),h(E))\right]$. The argument is similar to that of the previous case, noting that edges $r$ deletes are included in $E^\ominus$ and hence are never part of such a path in $G$.\\
		
		\noindent\emph{Case (3).} Suppose that $v,w\in V_K$ and $E_L\subset E_R$. Then $\text{LPath}(r,p)$ returns:
		\[ \mathtt{path}\mt{(}v,w,\mathtt{not}\ E\mt{)}\ \mathtt{or}\ \text{FuturePaths}(r,p). \]
		
		\noindent For paths along edges $e\in E_L,E_R$, the argument is as before. For paths including edges $e\in E_R\setminus E_L$, the equality then holds if:
		\[\left[\text{FuturePaths}(r,p)^{I,g}\ \text{iff}\ \text{path}_H(h(v),h(w),h(E))\right].\]
		
		\noindent By construction, $\text{FuturePaths}(r,p)$ expands to a disjunction of:
		\begin{center}
			\begin{tabular}{c}
				$\mt{(}\mathtt{path}\mt{(}v,x_1,\mathtt{not}\ E\mt{)}\ \mathtt{and}\ \mathtt{path}\mt{(}y_1,x_2,\mathtt{not}\ E\mt{)} \dots \mathtt{and}\ \mathtt{path}\mt{(}y_i,x_{i+1},\mathtt{not}\ E\mt{)}$\\
				$\dots \mathtt{and}\ \mathtt{path}\mt{(}y_n,w,\mathtt{not}\ E\mt{)}\mt{)}$
			\end{tabular}
		\end{center}

		\noindent over all non-empty sequences of distinct pairs $\langle \langle x_1,y_1 \rangle, \dots, \langle x_n,y_n \rangle \rangle$ drawn from:
		\[ \{\langle x,y \rangle \mid x,y\in V_K \wedge \text{path}_R(x,y,E) \wedge \neg \text{path}_L(x,y,E)\}. \]
		
		\noindent Assume that one such disjunct is satisfied in $G$. Then there are paths from $g(v)$ to $g(x_1)$, \dots $g(y_i)$ to $g(x_{i+1})$, \dots and $g(y_n)$ to $g(w)$ (all excluding edges in $g(E)$). These paths also exist in $H$, since each $x,y$ pair is in the interface of $r$, since $r$ does not delete edges, and since any nodes deleted would not have been part of the paths in $G$ (to satisfy the dangling condition, all incident edges must also be deleted, which would lead to a contradiction). Furthermore, each endpoint $h(x)$ of these paths is connected to the beginning of the next one $h(y)$ (since $\text{path}_R(x,y,E)$, and each element of the path is injectively mapped to $H$). Together, we have a witness for $\text{path}_H(h(v),h(w),h(E))$.
		
		Assume now that $\text{path}_H(h(v),h(w),h(E))$ holds by some path consisting of edges $e_1,e_2,\dots,e_n$ in sequence, with $s(e_1) = h(v)$, and $t(e_n) = h(w)$. Let this path be denoted by the pair $\langle h(v), h(w) \rangle$. Assume that at least one of its edges is in $h(E_R\setminus E_L)$, i.e.\ is created by $r$. The path contains path segments $e_a, \dots , e_b$ in $h(E_R\setminus E_L)$, denoted by $\langle s(e_a), t(e_b) \rangle$, such that if there is an edge $e_{a-1}$ (resp.\ $e_{b+1}$) in $\langle h(v), h(w) \rangle$, that edge is in the set $h(E_L)$. Since nodes created by $r$ can only be incident to edges in $E_R\setminus E_L$ (by the dangling condition), there exists in $G$ the same sequence of edges $\langle h(v), h(w) \rangle$ but with ``gaps'' for all such path segments $\langle s(e_a), t(e_b) \rangle$. Each pair of nodes $s(e_a), t(e_b)$ corresponds to a pair $x_i,y_i \in V_K$ such that $g(x_i) = h(x_i) = s(e_a)$ and $g(y_i) = h(y_i) = t(e_b)$. The construction returns a disjunct:		
		\[ \mt{(}\mathtt{path}\mt{(}v,x_1,\mathtt{not}\ E\mt{)}\ \mathtt{and}\ \mathtt{path}\mt{(}y_1,x_2,\mathtt{not}\ E\mt{)} \dots \mathtt{and}\ \mathtt{path}\mt{(}y_m,w,\mathtt{not}\ E\mt{)}\mt{)} \]
		
		\noindent for all such pairs $x_i,y_i$, since paths between them are created by $r$ (that is, $\text{path}_R(x_i,y_i,E)$ holds and $\text{path}_L(x_i,y_i,E)$ does not). As the disjunct evaluates to true under $I$ and $g$, so does $\text{FuturePaths}(r,p)$.\\

		\noindent\emph{Case (4).} Suppose that $v,w\in V_K$ and $E_L \neq E_R$, i.e.\ including the previous cases but also rules that both delete and create edges. Here, $\text{LPath}(r,p)$ returns:
		\[ \mathtt{path}\mt{(}v,w,\mathtt{not}\ E^\ominus\mt{)}\ \mathtt{or}\ \text{FuturePaths}(r,p) \]
		
		\noindent with $\text{FuturePaths}(r,p)$ expanding as before but with $E^\ominus$ replacing $E$. The argument is as in the previous case, but noting that edges $r$ deletes are included in $E^\ominus$, and hence are not considered in $H$ (nor in any corresponding path segments in $G$).\\
		
		\noindent\emph{Case (5).} Now, suppose that $v\notin V_K$, $w\in V_K$, and $E_L \neq E_R$ (the rule must create at least one edge from from $v$ in $V_R$). Then $\text{LPath}(r,p)$ returns the constraint:
		
		\begin{center}\begin{tabular}{c}
			$\mathtt{false\ or}\ \mathtt{path}\mt{(}x_1,w,\mathtt{not}\ E^\ominus\mt{)}\ \mathtt{or}\ \mathtt{path}\mt{(}x_2,w,\mathtt{not}\ E^\ominus\mt{)}\ \mathtt{or}\dots$\\
		 $\dots \mathtt{or}\ \text{FuturePaths}(r,p)$
		\end{tabular}\end{center}
		
		\noindent for each $x_i \in V_K$ such that $\text{path}_R(v,x_i,E^\ominus)$, and the equality holds if: \[\left[\text{LPath}(r,p)^{I,g}= \text{true} \ \text{iff}\ \text{path}_H(h(v),h(w),h(E^\ominus))\right].\]
		
		\noindent Assume that $\text{LPath}(r,p)^{I,g} = \text{true}$. Suppose that $\text{LPath}'(r,v,w,E^\ominus) = \text{true}$. Then from the construction, there exists some $z\in V_K$ with $\text{path}_R(v,z,E^\ominus)$ such that $(\mathtt{path}\mt{(}z,w,\mathtt{not}\ E^\ominus\mt{)})^{I,g}=\text{true}=\text{LPath}(r,\mathtt{path}\mt{(}z,w,\mathtt{not}\ E^\ominus\mt{)})^{I,g}$. By induction we get $\text{path}_H(h(z),h(w),h(E^\ominus))$. From the assumption $\text{path}_R(v,z,E^\ominus)$ and the definition of morphisms we derive that $\text{path}_H(h(v),h(z),E^\ominus)$, and hence the result that $\text{path}_H(h(v),h(w),h(E^\ominus)$. Suppose $\text{FuturePaths}(r,p)^{I,g} = \text{true}$ instead, i.e. 
		\[ \mt{(}\mathtt{path}\mt{(}z,x_1,\mathtt{not}\ E^\ominus\mt{)}\ \mathtt{and}\ \mathtt{path}\mt{(}y_1,x_2,\mathtt{not}\ E^\ominus\mt{)} \dots \mathtt{and}\ \mathtt{path}\mt{(}y_m,w,\mathtt{not}\ E^\ominus\mt{)}\mt{)} \]
		
		\noindent for some $z\in V_K$ such that $\text{path}_R(v,z,E^\ominus)$. Let $p' = \mathtt{path}\mt{(}z,w,\mathtt{not}\ E^\ominus\mt{)}$. Then:
		\[ \text{FuturePaths}(r,p)^{I,g} = \text{true} = \text{FuturePaths}(r,p')^{I,g} = \text{LPath}(r,p')^{I,g}.\]
		
		\noindent By induction, we have that $\text{path}_H(h(z),h(w),h(E^\ominus))$. As before we derive that $\text{path}_H(h(v),h(z),E^\ominus)$, and hence the result that $\text{path}_H(h(v),h(w),h(E^\ominus)$.

		Assume that $\text{path}_H(h(v),h(w),h(E^\ominus))$. There is a node $z\in V_K$ such that $\text{path}_R(v,z,E^\ominus)$ and $\text{path}_H(h(z),h(w),h(E^\ominus))$. (Suppose there is no such node. Then every node reachable from the image of $v$ in $H$ will also have been created by the rule, and in particular, none of these nodes will be the image of $w$ since $w\in V_K$. A contradiction.) Let $p' = \mathtt{path}\mt{(}z,w,\mathtt{not}\ E\mt{)}$. By induction, we have that:
		\[\text{LPath}'(r,p')^{I,g} = (\mathtt{path}\mt{(}z,w,\mathtt{not}\ E^\ominus\mt{)}\ \mathtt{or}\ \text{FuturePaths}(r,p'))^{I,g} =  \text{true}.\]
		
		\noindent If the first disjunct evaluates to true, then so does $\text{LPath}'(r,p)^{I,g}$ since the construction yields a disjunct $\mathtt{path}\mt{(}x_i,w,\mathtt{not}\ E^\ominus\mt{)}$ where $x_i = z$. If the second disjunct evaluates to true, i.e.\
		\[ \mt{(}\mathtt{path}\mt{(}z,x_1,\mathtt{not}\ E^\ominus\mt{)}\ \mathtt{and}\ \mathtt{path}\mt{(}y_1,x_2,\mathtt{not}\ E^\ominus\mt{)} \dots \mathtt{and}\ \mathtt{path}\mt{(}y_m,w,\mathtt{not}\ E^\ominus\mt{)}\mt{)} \]
		
		\noindent then so does $\text{LPath}'(r,p)^{I,g}$ since the above is yielded by the construction, i.e.
		\[\text{LPath}'(r,v,x_1,E^\ominus) = \mathtt{false}\ \mathtt{or}\ \mathtt{path}\mt{(}z,x_1,\mathtt{not}\ E^\ominus\mt{)}\ \mathtt{or}\dots. \]

		\noindent \emph{Case (6).} Suppose that $v\in V_K,w\notin V_K$, and $E_L \neq E_R$. Analogous to Case (5). \\
		
		\noindent \emph{Case (7).} Finally, suppose that $v,w \notin V_K$ and $E_L \neq E_R$. There are two subcases: (A) there are nodes in $V_K$ which are reachable from $v,w$ in $R$; and (B) there are no such nodes in $V_K$. For the former subcase, the proof is along the lines of Cases (5)-(6). For the latter subcase, a path can only exist in $H$ if there is a path from $v$ to $w$ in $R$. If there is, the construction returns the disjunct $\mathtt{true}$; otherwise $\mathtt{false}$. The equality clearly holds.
		\qed
	\end{proof}

	\newpage
	\subsection{Proof of Lemma \ref{lemma:L}}\label{lemma:L:PROOF}
	
	\begin{proof}
		By structural induction. Note that the construction distinguishes two cases in the step, according to whether a pushout complement exists or not. We will consider the case that one does; the other proceeds analogously to the proof of Theorem 6 in \cite{Habel-Pennemann09a}.\\
		
		\noindent\emph{Basis.} Let $c = \text{true}$. Then $\text{L}'(r,c,M) = \text{true}$. All morphisms satisfy true.\\
		
		\noindent\emph{Step. Only if.} Let $c = \exists_\mathtt{V}\mathtt{X}\mt{[}c'\mt{]}$. By assumption,
		\[g\models^I \text{L}'(r,\exists_\mathtt{V}\mathtt{X}\mt{[}c'\mt{]},M) = \exists_\mathtt{V}\mathtt{X}\mt{[}\ \bigvee_{M'\in 2^{M_\mathtt{V}}}\text{L}'(r,c',M\cup M')\ \mt{]}.\]
		
		\noindent There exists an $M'$ such that $g\models^I \exists_\mathtt{V}\mathtt{X}\mt{[} \text{L}'(r,c',M\cup M')\mt{]}$. By the definition of $\models$, there exists some $V \subseteq V_G$ such that  $I' = I \cup \{\mathtt{X}\mapsto V\}$ and $g \models^{I'} \text{L}'(r,c',M\cup M')$. By induction, we get that $h \models^{I'_{M\cup M'}} c'$. Observe that: 
		\[I'_{M\cup M'} = I_M \cup \{\mathtt{X}\mapsto V \cup \{h(x) \mid (x,\mathtt{X})\in M'\}\}.\]
		
		\noindent By definition of $\models$, we get the result that $h \models^{I_M} \exists_\mathtt{V}\mathtt{X}\mt{[}c'\mt{]}$. (Analogous for case $c = \exists_\mathtt{E}\mathtt{X}\mt{[}c'\mt{]}$.)
		
		Now let $c = \exists (a\mid \gamma, c')$. By assumption,
		\[g \models^I \text{L}'(r,\exists (a\mid \gamma, c'),M) = \exists(b\mid \gamma_M,\text{L}'(r^*,c',M))\]
		
		\noindent i.e.\ there is a morphism $q'\!:Y\hookrightarrow G$ such that $\gamma_M^{I,q'} = \text{true}$, $q'\circ b = g$, and $q' \models^I \text{L}'(r^*,c',M)$. Following the proof of Theorem 6 in \cite{Habel-Pennemann09a}, we derive a morphism $q\!: X\hookrightarrow H$ with $q \circ a = h$, i.e.\ $h \models^I \exists(a)$. Consider now the construction of $\gamma_M$ from $\gamma$. Each MSO expression $x\in\mathtt{X}$ for $x\in X \setminus Y$ is replaced in $\gamma_M$ by $\tt true$ (resp.\ $\tt false$) if $(y,\mathtt{X})\in M$ (resp.\ $\notin$) for some $y = x$. Observe that each $(x\in\mathtt{X})^{I_M,q}$ evaluates to the same Boolean value as the corresponding replacement ($\tt true$ or $\tt false$) in $\gamma_M$ under $I,q'$. Moreover, each path predicate $p$ in $\gamma$ is replaced with $\text{LPath}(r^*,p)$. By Proposition \ref{prop:LPath},
		\[ \text{LPath}(r^*,p)^{I,q'} = p^{I,q}. \]
		
		\noindent Together, we have $\gamma_M^{I,q'} = \gamma^{I_M,q} = \text{true}$ and $h \models^{I_M} \exists(a\mid\gamma)$. From the induction hypothesis, $q' \models^I \text{L}'(r^*,c',M)$ implies $q \models^{I_M} c'$. With this, we get the result that $h \models^{I_M} \exists (a \mid \gamma, c')$.\\

		\noindent\emph{Step. If.} Let $c = \exists_\mathtt{V}\mathtt{X}\mt{[}c'\mt{]}$. By assumption, $h \models^{I_M} \exists_\mathtt{V}\mathtt{X}\mt{[}c'\mt{]}$. By the definition of $\models$, there exists some $V\subseteq V_H$ such that $h \models^{I_M \cup \{\mathtt{X}\mapsto V\}} c'$. Define:
		\[M_\mathtt{X} = \{ (x,\mathtt{X}) \mid x\in V_R\setminus V_L \wedge h(x) \in V \}\]
		
		\noindent and $V^\ominus = V \setminus (V_H\setminus V_D)$. Now define:
		\[ I' = I \cup \{\mathtt{X}\mapsto V^\ominus\} \]
		
		\noindent and hence:
		\[ I'_{M\cup M_\mathtt{X}} = I_M \cup \{\mathtt{X}\mapsto V^\ominus \cup \{h(x) \mid (x,\mathtt{X})\in M_\mathtt{X}\} \}. \]
		
		\noindent Observe that $h \models^{I'_{M\cup M_\mathtt{X}}} c'$. By induction hypothesis, we get $g \models^{I'} \text{L}'(r,c',M\cup M_\mathtt{X})$. Clearly, $M_\mathtt{X}$ is in $2^{M_\mathtt{V}}$ from the construction, hence:
		\[g\models^{I'}  \bigvee_{M'\in 2^{M_\mathtt{V}}}\text{L}'(r,c',M\cup M').\]
		
		\noindent By the definition of $\models$, we get the result that:
		
		\[g\models^{I}  \exists_\mathtt{V}\mathtt{X}\mt{[}\ \bigvee_{M'\in 2^{M_\mathtt{V}}}\text{L}'(r,c',M\cup M')\ \mt{]} = \text{L}'(r,c,M).\]
		
		\noindent (Analogous for case $c = \exists_\mathtt{E}\mathtt{X}\mt{[}c'\mt{]}$.)
		
		Now let $c = \exists (a\mid \gamma, c')$. By assumption, $h \models^{I_M} \exists (a\mid \gamma, c')$, i.e.\ there is a morphism $q\!:X\hookrightarrow H$ with $\gamma^{I_M,q} = \text{true}$, $q\circ a = h$, and $q \models^{I_M} c'$. Following the proof of Theorem 6 in \cite{Habel-Pennemann09a}, we derive a morphism $q'\!: Y \hookrightarrow G$ with $q' \circ b = g$, i.e.\ $g \models^I \exists(b)$. Consider now the construction of $\gamma_M$ from $\gamma$. Each MSO expression $x\in\mathtt{X}$ for $x\in X \setminus Y$ is replaced in $\gamma_M$ by $\tt true$ (resp.\ $\tt false$) if $(y,\mathtt{X})\in M$ (resp.\ $\notin$) for some $y = x$. Observe that each $(x\in\mathtt{X})^{I_M,q}$ evaluates to the same Boolean value as the corresponding replacement ($\tt true$ or $\tt false$) in $\gamma_M$ under $I,q'$. Moreover, each path predicate $p$ in $\gamma$ is replaced with $\text{LPath}(r^*,p)$. By Proposition \ref{prop:LPath},
		\[ \text{LPath}(r^*,p)^{I,q'} = \text{LPath}(r^*,p)^{I_M,q'} = p^{I_M,q}. \]
		
		\noindent Together, we have $\gamma_M^{I,q'} = \gamma^{I_M,q} = \text{true}$ and $g \models^{I} \exists(b\mid\gamma_M)$. From the induction hypothesis, $q \models^{I_M} c'$ implies $q' \models^I \text{L}'(r^*,c',M)$. Finally, we get the result that:
		\[ g \models^I \exists (b\mid\gamma_M, \text{L}'(r^*,c',M)) = \text{L}'(r,c,M).  \]

		For Boolean formulae over M-conditions, the statement follows from the definition of $\models$ and the induction hypothesis.
		\qed
	\end{proof}

\end{document}